%% file: main.tex
\documentclass{article}

\PassOptionsToPackage{numbers}{natbib}
\usepackage[preprint]{neurips_2024}

\usepackage{graphicx}
\usepackage{subfigure}

\usepackage[utf8]{inputenc} 
\usepackage[T1]{fontenc}    
\usepackage{hyperref}       
\usepackage{url}            
\usepackage{booktabs}       
\usepackage{amsfonts}       
\usepackage{nicefrac}       
\usepackage{microtype}      
\usepackage{xcolor}         

\usepackage{amsmath} 
\usepackage{algorithm}
\usepackage{algorithmic}

\usepackage{amsmath}
\usepackage{amssymb}
\usepackage{mathtools}
\usepackage{amsthm}
\usepackage{multirow}
\usepackage[capitalize,noabbrev]{cleveref}
\theoremstyle{plain}
\newtheorem{theorem}{Theorem}[section]

\newtheorem{lemma}[theorem]{Lemma}
\newtheorem{corollary}[theorem]{Corollary}
\theoremstyle{definition}

\newtheorem{assumption}[theorem]{Assumption}
\theoremstyle{remark}

\title{Modeling Other Players with Bayesian Beliefs for Games with Incomplete Information}

\author{
    Zuyuan Zhang\\
    The George Washington University\\
    \texttt{zuyuan.zhang@gwu.edu}\\
    \And
    Mahdi Imani\\
    Northeastern University\\
    \texttt{m.imani@northeastern.edu}
    \And
    Tian Lan\\
    The George Washington University\\
    \texttt{tlan@gwu.edu}
}

\begin{document}

\maketitle

\input{01abstract}

\input{02introduction}

\input{03background}

\input{04solution}

\input{05experiment}

\input{06Conclusion}

{
\small
\bibliography{ref}
\bibliographystyle{unsrtnat} 
}








\input{07appendix}

\end{document}

%% file: 01abstract.tex
\begin{abstract}
Bayesian games model interactive decision-making where players have incomplete information -- e.g., regarding payoffs and private data on players’ strategies and preferences -- and must actively reason and update their belief models (with regard to such information) using observation and interaction history. Existing work on counterfactual regret minimization have shown great success for games with complete or imperfect information, but not for Bayesian games. To this end, we introduced a new CFR algorithm: \textbf{Bayesian-CFR} and analyze
its regret bound with respect to Bayesian Nash Equilibria in Bayesian games.
First, we present a method for updating the posterior distribution of beliefs about the game and about other players’ types. The method uses a kernel-density estimate and is shown to converge to the true distribution. 
Second, we define Bayesian regret and  present a Bayesian-CFR minimization algorithm for computing the Bayesian Nash equilibrium. 
Finally, we extend this new approach to other existing algorithms, such as Bayesian-CFR+ and Deep Bayesian CFR. 
Experimental results show that our proposed solutions significantly outperform existing methods in classical Texas Hold'em games. 
\end{abstract}

%% file: 02introduction.tex
\section{Introduction}
Bayesian Games, also known as Games with Incomplete Information~\cite{zamir2020bayesian}, are models of interactive decision situations, where the players (i.e., decision-making agents) have only partial information about the data of the game and about the other players, such as the state of the game, payoffs, and private information on players' strategies and preferences. It provides a power model for strategic interaction, as the overwhelming
majority of real-life situations involve players that have only partial information and must reason about behaviors, goals, and beliefs of the other players and autonomous agents~\cite{Albrecht_2018}. Bayesian Games are 
widely employed in various real-world domains, such as gaming~\cite{zamir2020bayesian}, transportation planning~\cite{bernhard2019addressing}, economics~\cite{george1970market}, and networked systems~\cite{maccarone2021stochastic,asheralieva2021fast}.

Counterfactual Regret Minimization (CFR)~\cite{DBLP:conf/nips/ZinkevichJBP07} is a state-of-the-art approach for computing a Nash equilibrium. 
The idea of CFR is to compute or approximate a Nash equilibrium by updating strategies based on the cumulative regret calculated among different agents. It has also been extended to large and complex games using efficient sampling~\cite{lanctot2009monte}, deep neural networks~\cite{brown2019deep}, and hierarchical policies~\cite{chen2023hierarchical}. CFR has been successfully applied to Complete Information Games (where data of the game and other players are common knowledge) and Imperfect Information Games (where players are
simply unaware of the actions chosen by others)~\cite{DBLP:conf/nips/ZinkevichJBP07,chen2023hierarchical}. In Bayesian Games, the use of CFR to compute Bayesian Nash Equilibria (BNE)~\cite{ferguson2020course} -- which requires actively reasoning and updating belief models of other players -- has not been considered. 

In this paper, we propose Bayesian Counterfactual Regret (Bayesian-CFR) minimization and analyze its regret bound with respect to Bayesian Nash Equilibria. In particular, players maintain beliefs about the game and about other players' types~\cite{zamir2020bayesian} and update their beliefs throughout the extensive-form game. Thus, players have the ability to reason about others by constructing belief models of the other players. The belief model of player $\chi$ is a function that takes as input some locally observed interaction history $O_\chi$ and returns a prediction $Pr_\chi(\theta|O_\chi)$ of some property of interest (denoted by $\Theta$) regarding the modelled players, such as private information on their payoffs,  strategies, and preferences. Our goal is to compute BNE. It is defined as a strategy profile that maximizes the expected payoff for each player given their beliefs and given the strategies played by the other players~\cite{zamir2020bayesian}. Theoretical analysis of Bayesian games often use the seminal Harsanyi's Model~\cite{zamir2020bayesian}. However, Bayesian-CFR minimization for computing BNE in practical problems like Texas hold'em~\cite{southey2012bayes} has not been considered. CFR-based methods and regret bound analysis cannot be directly applied due to the evolving beliefs (and their combinatorial dynamics) in the extensive-form game.

We develop a Bayesian-CFR minimization algorithm for computing BNE. By formulating an immediate Bayesian counterfactual regret, we prove that minimizing such immediate Bayesian counterfactual regret can also minimize overall Bayesian regret. The algorithm leverages a kernel-density approach to recursively and efficiently update the posterior belief distribution of the players, using their local history. The posterior is shown to result in an unbiased estimate of the true distribution. Next, for complex and large Bayesian games, e.g., due to growing number of opponents~\cite{harsanyi1967games}, inspired by CFR+~\cite{tammelin2014solving} and Deep CFR~\cite{brown2019deep}, we further propose Bayesian CFR+ and Bayesian Deep-CFR, by considering the Bayesian cumulative counterfactual regret and leveraging deep neural network representations, respectively. The proposed Bayesian-CFR minimization algorithm is guaranteed to compute BNE of the Bayesian game, with a new regret bound depending on an extra term $\Delta_{\Theta}^T$ with respect to time $T$ and the dimension of belief models ${\Theta}$. 

The proposed solutions are evaluated on a Texas Hold'em poker environment, a popular and real-life game, by evaluating the resulting exploitability of our algorithms against state-of-the-art baselines, including CFR~\cite{zinkevich2007regret}, CFR+~\cite{tammelin2014solving}, Deep-CFR~\cite{brown2019deep}, MCCFR~\cite{lanctot2009monte}, and DQN~\cite{li2018double}. The experiment results, using a wide range of player behaviors (e.g., aggressive, neutral, conservative, pure and mixed player types) demonstrate the superior performance of the proposed Bayesian-CFR framework.

%% file: 03background.tex
\section{Background}

\subsection{Extensive-Form Bayesian Games}

Bayesian Games are strategic decision-making models where players lack complete knowledge about some elements of the game, such as payoffs or private information of other players. In contrast to Imperfect Information Games~\cite{ferguson2020course,DBLP:conf/nips/ZinkevichJBP07}, where players simply may not know the actions chosen by others but understand the data of the game and the players, Bayesian games further relax the
assumption, to enable modeling of the overwhelming majority of real-life situations, where players have only partial information about the payoff relevant data of the game and the other players. Theoretical analysis of Bayesian games includes~\cite{harsanyi1967games}. A comprehensive survey of Bayesian games can be found in~\cite{dekel2004learning}.

Extensive-form games~\cite{shoham2008multiagent} are a tree-based method used to describe games with incomplete information.  An extensive-form game can be described as a tuple $\Gamma = <\mathcal{N},H,Z,P,u,I,\sigma_{c}>$,
where $\mathcal{N} = \{1,...,N\}$ is a set of players, $H$ is a set of histories. The root node of the game tree in $H$ is the empty sequence $\emptyset$, and every sequence prefix in $H$ is also in $H$. Given $h,h^{'}\in H$, we write $h^{'} \sqsubseteq h$ if $h^{'} $is a prefix of $h$.
$Z\subset H$ is the set of the terminal histories, and $A(h)=\{a:(h,a)\in H\}$ is the set of available actions at a non-terminal history $h\in H$. $P$ is the player function where $P(h)$ is the player who takes an action at the history $h$. This function can be expressed as $P(h) \rightarrow N \cup\{c\}$, where
$c$ denotes the chance player, which represents a stochastic event outside of the players' control. If $P(h) = c$, then chance determines the action taken at history $h$.
Information sets $I_{i}\in \mathcal{I}_{i} $ corresponds to one decision point of player $i$ which means that $P(h_{1}) = P(h_{2})$ and $A(h_{1}) = A(h_{2})$ of any $h_{1},h_{2}\in I_{i}$.
For convenience, we use $A(I_{i})$ to represent the set $A(h)$ and $P(I_{i})$ to represent the player $P(h)$ for any $h\in I_{i}$. For each player $i\in N$, a utility function is a mapping $u_{i}
:Z \rightarrow \mathbb{R}$. 

In this paper, we consider extensive-form Bayesian games, which are modeled as 
$\Gamma = \langle N, H, Z, Pr, u, I, \sigma_c, \Theta, Pr \rangle$. Here $\Theta$ represents a set of possible types of players (e.g., encapsulating their payoffs, strategies, and preferences), and $Pr(\theta)$ for $\theta\in\Theta$ denotes the players' true distribution over the type space. During the game, each player $\chi$ conditions on her local observations/interaction history $O_\chi$ to update her belief $Pr_\chi(\theta|O_\chi)$. The utility function is defined such that for each player $i \in N$, a utility function is a mapping $u_{i,\theta} : Z \rightarrow \mathbb{R}$.
A player's behavior strategy $\sigma_i$ is a function mapping every information set of player $i$ to a probability distribution over $A(I_{i})$.
A strategy profile $\sigma$ consists of a strategy for each player $\sigma_1, ..., \sigma_n$ with $\sigma_{-i}$ referring to all the strategies in $\sigma$ except $\sigma_{i}$.
In Bayesian games, the strategy $\sigma_{\Theta}$ is denoted as $\sigma_{1,\theta}, ..., \sigma_{n,\theta}$.
Let $\pi^{\sigma_{\Theta}}(h)$ be the reaching probability of history $h$ if players choose actions according to $\sigma_{\Theta}$.
Given a strategy profile $\sigma_{\Theta}$, the overall value to player $i$ is the expected payoff 
of the resulting terminal node, i.e.,
\begin{eqnarray}
    u_{i}(\sigma_{\Theta}) = \sum_{\theta\in \Theta}Pr_i(\theta)\sum_{h\in Z} \pi^{\sigma_{\Theta}}(h)u_{i,\theta}(h).
\end{eqnarray}
Bayesian Nash Equilibrium (BNE) is used as a solution to the problem. It refers to a strategy configuration where no party can gain a higher expected payoff by unilaterally changing their strategy, given their beliefs and given the strategies played by the other players. Formally, we can define BNE 
as follows:
\begin{equation}
\label{eql:1}
    u_{i,\Theta}(\sigma_{\Theta}) \geq \max_{\sigma^{'}_{\Theta}\in\sum_{i,\Theta}} u_{i,\Theta}(\sigma^{'}_{i,\Theta},\sigma_{-i,\Theta}), 
\end{equation}
where we use subscript $\Theta$ to denote the values with respect to posterior distribution of the belief model.
A near-BAE or $\epsilon$-Bayesian Nash equilibrium is a strategy profile $\sigma$, where it is not possible for any player to gain more than $\epsilon$ in expected payoff, i.e.,
\begin{equation}
    u_{i,\Theta}(\sigma_{\Theta}) + \epsilon \geq \max_{\sigma^{'}_{\Theta}\in\sum_{i,\Theta}} u_{i,\Theta}(\sigma^{'}_{i,\Theta},\sigma_{-i,\Theta}). 
\end{equation}

\subsection{Counterfactual Regret Minimization in imperfect game}
{CFR} is a family of iterative algorithms that are the most popular approach to approximately solving NE of large and complex games~\cite{zinkevich2007regret,lanctot2009monte,brown2019deep,chen2023hierarchical}.
This algorithm is mainly used in repeated games or extensive games, either with complete information or with imperfect information (i.e., where players are simply unaware of the actions chosen by others).
Let $\sigma^{t}_i$ be the strategy used by player $i$ on round $t$. 
Define $u_{i}(\sigma,h)$ as the expected utility of player $i$ given that the history $h$ is reached, and all players play according to strategy $\sigma$ from that point on.
Let also $u_i(\sigma,h\cdot a)$ be the expected utility of player $i$ given the history $h$, where all players play according to strategy $\sigma$ except for the $i$th player who selects action $a$ in the history $h$.
Formally, $u_{i}(\sigma,h) = \sum_{z\in Z} \pi^{\sigma}(h\cdot z)u_{i}(z)$ and $u_{i}(\sigma,h\cdot a) = \sum_{z\in Z} \pi^{\sigma}(h \cdot a,z)u_{i}(z)$.

The \textit{counterfactual value} $u_{i}^{\sigma}(I)$ for imperfect information games is the expected value of an information set $I$ given that player $i$ tries to reach it. This value is the weighted average of the value of each history in an information set. The weight is proportional to the contribution of all players other than $i$ to reach each history. Thus, $u_{i}^{\sigma}(I) = \sum_{h\in I} \pi_{-i}^{\sigma}(h)\sum_{z\in Z} \pi^{\sigma}(h \cdot z)u_{i}(z)$. For any action $a\in A(I)$, the counterfactual value of an action $a$ is $u_{i}^{\sigma}(I,a) = \sum_{h\in I} \pi_{-i}^{\sigma}(h)\sum_{z\in Z} \pi^{\sigma}(h \cdot a, z)u_{i}(z)$.
The \textit{instantaneous regret} for action $a$ in information set $I$ on iteration $t$ is $r^{t}(I,a) = u_{P(I)}^{\sigma^{t}}(I,a) - u_{P(I)}^{\sigma^{t}}(I)$ 
{where P(I) is the player function mapping to player from information set $I$.}

The cumulative regret for action $a$ in $I$ on iteration $T$ is $R^{T}(I,a) = \sum_{t=1}^{T}r^{t}(I,a)$.
In CFR, players use \textit{Regret Matching} to pick a distribution over actions in an information set in proportion to the positive cumulative regret on those actions. 
Formally, on iteration $T$, player $i$ selects actions $a\in A(I)$ according to probability:
\begin{equation}
\label{eqn:strategy}
\sigma^{T}(I,a) = 
\left\{
\begin{aligned}
    &\frac{R_{+}^{T-1}(I,a)}{\sum_{b\in A(I)}R_{+}^{T-1}(I,b)} & if \sum_{b\in A(I)} R_{+}^{T-1}(I,b) > 0\\
    &\frac{1}{\lvert A(I) \rvert} & otherwise
    \end{aligned}
\right.
\end{equation}

where $R_{+}^{T}(I,a) = \max(R^{T}(I,a),0)$ which can sets probabilities proportional to the positive regrets: $\sigma^{T}(I,a)\propto \max(R^{T}(I,a),0)$.
If a player updates strategy according to CFR in every iteration, then on iteration $T$, $R^{T}(I) \leq \Delta_{i}\sqrt{\lvert A_{i} \rvert}\sqrt{T}$ where $\Delta_{i}=\max_{z}u_{i}(z) -\min_{z}u_{i}(z)$ is the bound of utility of player $i$.
Moreover, $R_{i}^{T} \leq \sum_{I\in \mathcal{I}_{i}}R^{T}(I)\leq \lvert \mathcal{I}_{i} \rvert \Delta_{i} \sqrt{\lvert A_{i} \rvert} \sqrt{T}$. Therefore, as $T \rightarrow \infty$, $\frac{R_{i}^{T}}{T} \rightarrow 0$.
In  zero-sum games, if all players' average regret $\frac{R_{i}^{T}}{T} \leq \epsilon$, CFR guarantees that then their average strategies $(\bar{\sigma}_{1}^{T},\bar{\sigma}_{2}^{T})$ form a $2\epsilon$-equilibrium~\cite{zinkevich2007regret} as $T\rightarrow \infty$.

%% file: 04solution.tex
\section{Bayesian-CFR}

Players in Bayesian games need to form and maintain local beliefs about the unknown properties/payoffs and reason about other players' behaviors~\cite{zamir2020bayesian}, throughout the extensive-form game. 
This necessitates reasoning from a Bayesian perspective and adjusting strategies accordingly, mirroring real-world scenarios of strategic interaction under uncertainty. A comprehensive survey on autonomous agents modelling other agents is provided in~\cite{Albrecht_2018}. In this paper, we propose a novel framework for Bayesian-CFR. It leverages a kernel-density posterior to update the players' belief models $Pr_{\chi}(\theta|O_{\chi})$ from local observations/interaction history $O_{\chi}$. The method is shown to yield an unbiased estimate of the true distribution $Pr_\chi$.  Then, we define a Bayesian average overall regret $R_{i,\Theta}^T$ over players' beliefs. Next, we show that minimizing immediate Bayesian counterfactual regret can also minimize overall Bayesian regret and thus propose a Bayesian-CFR algorithm that is proven to compute the BNE.  Finally, we also extend the proposed algorithm to Bayesian CFR+ and Bayesian Deep-CFR, by considering the Bayesian cumulative counterfactual regret and leveraging deep neural network representations, respectively. Regret bound of the proposed deep Bayesian-CFR is characterized. Due to space limitation, we collect proofs of all theorems and lemmas in the appendix, while provide some brief discussions. Pseudo code of all proposed algorithms, including Bayesian-CFR, Bayesian-CFR+, and Deep Bayesian-CFR are also included in the appendix.

\subsection{Update Posterior: Evidence from the other players}

We provide details on our proposed method for posterior belief update. At any stage of the Bayesian game, we denote the observations/interactions history of player ${\chi}$ by $O_{\chi} = \{h_1, h_2, \ldots, h_n\}$ in $n$ previous stages. 
We make the following assumptions about players' behavior:
Each player's type remains invariant with respect to time and does not change in the middle of a game. Each player attempts to maximize the total accumulated payoff according to some utility, given the belief of the game and of the other players' types. These are standard assumptions in Bayesian games. In real-life problems, players strategies often do not change drastically in the short term, and they act to maximize their own payoffs. Since each $h_i$ only depends on the current state of the game and type $\theta$, the posterior belief $Pr_\chi(\theta|O_{\chi})$ can be computed by:
\begin{equation}
\begin{aligned}
Pr_{\chi}(\theta|O_{\chi}) \propto P({\theta})  P(O_\chi|\theta) 
= P({\theta})  \prod_{i=1}^n  P(h_i|\theta),
\end{aligned}
\end{equation}
where $P({\theta})$ is a prior distribution and $P(h_i|\theta)$ can be calculated using the corresponding state of the game-tree (for CFR computation) for given $\theta$.

To make this belief update more efficient, we propose a kernel-density method for estimating the posterior probabilities. Inspired by~\cite{van2000asymptotic}, this is a non-parametric method and sample-based. Let $(h_j^{'},  \theta_j^{'})_{j=1}^m$ be a set of $m$ samples of observations/interactions given different types $\theta_j^{'}$ (which are collected from the game-tree). The idea is that under different $\theta\in\Theta$, the observations/interactions would also vary. 
Thus, we can use a conditional kernel density estimator (CKDE) with the $m$ samples as reference points:
\begin{equation}
\label{eqn:likelihood}
\begin{aligned}
\widehat{Pr}_\chi(h_i|\theta) = \frac{\widehat{Pr}_{\chi}(h_i,\theta)}{\widehat{Pr}_{\chi}(\theta)} = \sum\limits_{j=1}^{m}\frac{K(\frac{d_{s}(h,h_{j})}{w})K^{'}(\frac{d_r{(\theta,\theta_j})}{w^{'}})}{\sum_{l=1}^{m}K(\frac{d_r{(\theta,\theta_l})}{w^{'}})},
\end{aligned}
\end{equation}
where $K$ and $K^{'}$ are kernel functions with bandwidth $w,w^{'}>0$, and $d_{s}$ and $d_{r}$ are the distance function between histories and type distributions, respectively.
We compute the posterior belief estimate by applying the Bayes rule:
\begin{equation}
\label{eqn:bayesion}
\begin{aligned}
\widehat{Pr}_{m}^{n}(\theta|O_\chi) \propto P(\theta)\prod\limits_{i=1}^{n}\widehat{Pr}(h_i|\theta) = Pr(\theta)\prod\limits_{i=1}^{n}\sum\limits_{j=1}^{m}\frac{K(\frac{d_{s}(h_{i},h_{j})}{w})K^{'}(\frac{d_r{(\theta,\theta_j})}{w^{'}})}{\sum_{l=1}^{m}K(\frac{d_r{(\theta,\theta_l})}{w^{'}})}.
\end{aligned}
\end{equation}

To demonstrate the correctness of the {estimated posterior} algorithm, we provide the following proof to show that the CKDE method can converge to an unbiased estimate. First, we present the following assumptions:

\begin{assumption}
\label{ass:equal}
The distributions are considered equal if the likelihood of the observed values is the same under both types: $\theta_{1} \simeq \theta_{2}$ iff $||Pr(\cdot|\theta_{1}) - Pr(\cdot|\theta_{2})||_{L_{1}} = 0$
\end{assumption}

Based on assumption~\ref{ass:equal}, we can define the equivalence class $[\theta^{*}] = \{\theta : \theta \simeq \theta^{*}\}$ for the true player-type distribution $\theta^{*}$.

\begin{assumption}
$Pr(h|\theta)$ and $P(\theta)$ are square-integrable and twice differentiable with a square-integrable and continuous second-order derivative.
\end{assumption}

 \begin{lemma}\cite{mandyam2023kernel}
 \label{lem:1}
 Let $w_{m},w^{'}_{m} > 0$ be the bandwidths chosen to estimate the joint probability and marginal probability, respectively. with $mw_{m}^{d/2}\rightarrow \infty$ and $mw_{m}^{',d^{'}/2}\rightarrow \infty$ as $m\rightarrow\infty$, where $d$ is the dimension of $(h,\theta)$ and $d^{'}$ is the dimension of $\theta$. Then,
 \begin{equation}
     \lim\limits_{m\rightarrow \infty}\widehat{Pr}_{\chi}(O_{\chi}|\theta) = Pr_{\chi}(O_{\chi}|\theta)
 \end{equation}
 \end{lemma}
The above Lemma~\ref{lem:1} demonstrates the validity of CKDE, which means that with enough samples, our estimate of the posterior distribution becomes infinitely close to the true posterior distribution.
With this Lemma, we next present a Theorem explaining that when both $n$ and $m$ are sufficiently large, the posterior distribution can converge to the equivalence class of the true distribution.

\begin{assumption}
\label{ass:4}
Assume the prior for $\theta$, satisfies $Pr(\{\theta:\text{KL}(\theta^{*},\theta)<\epsilon\})>0$ for any $\epsilon > 0$, where KL is the Kullback–Leibler divergence.
\end{assumption}
\begin{theorem}
\label{the:bayesian_converge}
If assumption~\ref{ass:4} holds, and the posterior measure corresponding to the posterior density function $\widehat{Pr}_{m}^{n}$ according to Eqn.~\ref{eqn:bayesion}, denoted by $Pr_{m}^{n}$, is consistent w.r.t. the $L_1$ distance; that is,
\begin{equation}
    \lim\limits_{m\rightarrow\infty\atop n\rightarrow\infty}Pr_{m}^{n}(\{\theta:||Pr(\cdot|\theta)-Pr(\cdot|\theta^{*})||_{L_{1}}<\epsilon\}) = 1
\end{equation}
\end{theorem}

The aforementioned Theorem~\ref{the:bayesian_converge} indicates that, with enough samples, the formula can converge to the equivalence class of true type distribution $\theta^{*}$.

\subsection{Bayesian-CFR Minimization}
For extensive-form games, one of the most popular algorithms for computing NE is based on the CFR algorithm, which calculates strategies based on cumulative regret values through regret matching. In other words, the probability of each action is related to the cumulative regret value of that action. Therefore, to compute each agent's strategy, we need to know the cumulative regret values of their actions explicitly. In Bayesian games, we redefine regret:
The \textbf{Bayesian average overall regret} of player $i$ given current belief $Pr(\theta|O_{\chi})$ and action $a$ in $I$ on iteration $T$ is:
\begin{equation}
  R_{i,\Theta}^{T} = \frac{1}{T} \max\limits_{\sigma_{i,\theta}^{*}\in\sum_{i}}\sum\limits_{\theta\in \Theta}Pr(\theta|O_{\chi})\sum\limits_{t=1}^{T}(u_{i,\theta}(\sigma_{i,\theta}^{*},\sigma_{-i,\theta}^{t})-u_{i,\theta}(\sigma_{\theta}^{t})).  
\end{equation}
The posterior probability in the aforementioned formula often cannot be obtained on time. Moreover, calculating regret after obtaining the posterior probability usually results in lower efficiency.
Therefore, the aforementioned Theorem~\ref{the:bayesian_converge} can be expressed in the following equivalent form:
\begin{equation}
  R_{i,\Theta}^{T} = \frac{1}{T} \max\limits_{\sigma_{i,\theta}^{*}\in\sum_{i}}\sum\limits_{t=1}^{T}\sum\limits_{\theta\in \Theta}Pr(\theta|O_{\chi}^{t})(u_{i,\theta}(\sigma_{i,\theta}^{*},\sigma_{-i,\theta}^{t})-u_{i,\theta}(\sigma_{\theta}^{t})).  
\end{equation}
{In this scenario, the posterior distribution is conditioned on the latest data up to time $t$, allowing for the consideration of uncertainty over time in the computation of Bayesian regret. }

The basic idea of our method is to decompose the overall regret into a set of additional regret terms that can be independently minimized. We have extended counterfactual regret to Bayesian counterfactual regret, defined on individual information sets and across each distribution of types. We demonstrate that the sum of Bayesian counterfactual regrets constrains the overall Bayesian regret, allowing for independently minimizing the Bayesian counterfactual regret for each type at every information set.

The \textbf{Bayesian immediate counterfactual regret} is:
\begin{equation}
\label{eqn:regret_imm}
  R_{i,\Theta,imm}^{T}(I) = \frac{1}{T} \max\limits_{a\in A(I)}\sum\limits_{t=1}^{T}\sum\limits_{\theta\in \Theta}Pr(\theta|O^{t}_{\chi})\pi_{-i,\theta}^{\sigma^{t}}(I)(u_{i,\theta}(\sigma_{\theta}^{t}|_{I\rightarrow a})-u_{i,\theta}(\sigma_{\theta}^{t},I)).  
\end{equation}
The above expression represents the regret value of a player making decisions at the information set $I$ in terms of Bayesian counterfactual utility, given knowledge of the player's distribution.
Because negative regret often indicates the selection of better action, we typically focus on the regret produced by poor choices made by the current strategy. Therefore, let $R_{i,\Theta,imm}^{T,+}(I) = \max(R_{i,\Theta,imm}^{T}(I),0)$ be the positive portion of immediate counterfactual regret.

\begin{theorem}
\label{the:regret}
$R_{i,\Theta}^{T}\leq \sum_{I\in \mathcal{I}_{i}}R_{i,\Theta,imm}^{T,+}(I)$ holds.
\end{theorem}

The theorem~\ref{the:regret} shows that minimizing immediate Bayesian counterfactual regret can also minimize overall Bayesian regret. Therefore, if we can only minimize immediate Bayesian counterfactual regret, we can find an approximate Nash equilibrium.

Like CFR, the Bayesian immediate counterfactual regret can also be minimized independently at each information set using approximation algorithms. We can obtain:
\begin{equation}
\label{eqn:regret_I}
  R_{i,\Theta}^{T}(I,a) = \frac{1}{T} \sum\limits_{t=1}^{T}\sum\limits_{\theta\in \Theta}Pr_{i}(\theta|O^{t}_{\chi})\pi_{-i,\theta}^{\sigma^{t}}(I)(u_{i,\theta}(\sigma_{\theta}^{t}|_{I\rightarrow a})-u_{i,\theta}(\sigma_{\theta}^{t},I)).  
\end{equation}

Therefore, we can propose the following new strategy update method:
\begin{equation}
\label{eqn:bayesian_strategy}
\sigma_{\Theta}^{T+1}(I,a) = 
\left\{
\begin{aligned}
    &\frac{R_{+,\Theta}^{T}(I,a)}{\sum_{b\in A(I)}R_{+,\Theta}^{T}(I,b)} & if \sum_{b\in A(I)} R_{+,\Theta}^{T}(I,b) > 0\\
    &\frac{1}{\lvert A(I) \rvert} & otherwise
    \end{aligned}
\right.
\end{equation}
{where $R_{+,\Theta}^{T}(I,a) = \max(R_{i,\Theta}^{T}(I,a),0)$.} To demonstrate the correctness of the aforementioned update strategy, we present the following theorem. It states that within a finite time, the total regret can always converge to an appropriate bound:

\begin{theorem}
\label{the:bound1}
If player $i$ selects actions according to Eqn.~\ref{eqn:bayesian_strategy} then $R_{i,\Theta,imm}^{T}(I)\leq \Delta_{u,i,\Theta}^{T}\sqrt{|A_i|}/\sqrt{T}$ and consequently $R_{i,\Theta}^{T}\leq\Delta_{u,i,\Theta}^{T}|\mathcal{I}_{i}|\sqrt{|A_i|}/\sqrt{T}$, where $|A_{i}|=\max_{h:P(h)=i}|A(h)|$, and $\Delta_{u,i,\Theta}=\sum_{\theta\in\Theta}Pr(\theta|O_{\chi}^{T})\max_{a,a^{'}\in A}(u_{T,\theta}(a)-u_{T,\theta}(a^{'}))$
\end{theorem}
Theorem~\ref{the:bound1} states that the described strategy can be used to compute Nash equilibria for self-play. The average overall regret bound is influenced by the player's probability distribution. However, as the algorithm progresses, the estimated probability distribution of the player becomes more accurate, and this bound will gradually tighten.

The pseudo-code for the algorithm can be seen in the Appendix.

\section{Extend Bayesian Counterfactual Regret Minimization to other algorithms}

This proposed Bayesian-CFR algorithm can be easily generalized and incorporated into other CFR-based algorithms. We propose Bayesian CFR+ and Deep Bayesian-CFR. The regret bound for Deep Bayesian-CFR is quantified.
The pseudo-code of the algorithms is detailed in the Appendix.

\subsection{Bayesian CFR+}

CFR+ is typically an order of magnitude or more higher in computation time compared to previously known algorithms, while potentially requiring less memory. CFR+ is a vector form of an alternating update algorithm. CFR+ replaces the regret-matching algorithm used in CFR with a new algorithm, regret-matching+.
When extended to CFR+~\cite{tammelin2014solving}, it results in the \textbf{Bayesian cumulative counterfactual regret+}.
\begin{equation}
\label{eqn:cfr_plus}
\begin{aligned}
R_{i,\Theta}^{+,T}(I,a) = 
\left\{
\begin{aligned}
    &\max\{\sum_{\theta\in\Theta}Pr(\theta|O_{\chi}^{T})(u_{i,\theta}(\sigma_{I\rightarrow a}^{T},I)-u_{i,\theta}(\sigma_{\Theta}^{T,I})),0\}& T = 1\\
    &\max\{R_{i,\Theta}^{+,T-1}(I,a)+\sum_{\theta\in\Theta}Pr(\theta|O_{\chi}^{T})(u_{i,\theta}(\sigma_{I\rightarrow a}^{T},I)-u_{i,\theta}(\sigma_{\Theta}^{T,I})),0\} & otherwise
    \end{aligned}
\right.
\end{aligned}
\end{equation}
In CFR+, faster convergence to Nash equilibrium can be achieved. The effectiveness of the Bayesian module, demonstrating its differences from traditional algorithms, is more evident in the experiments discussed later in the text. Additionally, it exhibits less variability.

\subsection{Bayesian Deep CFR}
Similarly, Bayesian Counterfactual Regret can also be extended to deep neural network representations, similar to Deep CFR~\cite{brown2019deep}.
We use the following network structure: each player's advantage value network and the network approximates the final average strategy. The difference with Deep CFR is that we encode the player types using one-hot encoding as a separate module in our players' advantage value networks. This module is then merged with the second layer of the original advantage value network, forming an advantage network with type-specific features.
Below, we obtain an upper bound of regret obtained using Bayesian Counterfactual Regret in a Bayesian game. Let $T$ denote the number of BCFR iterations, $|A|$ the maximum number of actions at any infoset, and $K$ the number of traversals per iteration. Let $L_{\mathcal{R}}^{t}$ be the average MSE loss for $\mathcal{R}_{p}(I,a|\theta^{t})$ on a sample in $M_{r,p}$ at iteration $t$, and let $L_{\mathcal{R}^{*}}^{t}$ be the minimum loss achievable for any function $\mathcal{R}$. The result is stated in the following theorem.
\begin{theorem}
\label{the:deep_bound}
Let $L_{\mathcal{R}}^{t} - L_{\mathcal{R}^{*}}^{t} \leq \epsilon_{L}$.
If the value memories are sufficiently large and Eq.~\ref{eqn:regret_I} holds, then with probability $1-\rho$ total regret of player $p$ at time $T$ is bounded by
\begin{equation}
    R_{p,\Theta}^{T}\leq \Bigg(1+\frac{\sqrt{2}}{\sqrt{\rho K}} \Delta_{\Theta}^{T}|I_p|\sqrt{|A|}\sqrt{T} + 4T|I_{p}|\sqrt{|A|\Delta_{\Theta}^{T} \epsilon_{L}}\Bigg)
\end{equation}
As $T \rightarrow \infty$, the average regret $\frac{R_{p,\Theta}^{T}}{T}$ is bounded by $4|\mathcal{I}_{p}|\sqrt{|A|\Delta^{T}_{\Theta} \epsilon_{L}}$ with high probability, where $\Delta_{\Theta} = \sum_{\theta\in\Theta}Pr(\theta|O^{T}_{\chi})\Delta_{\theta}$.
\end{theorem}

The proof is inspired by that of Deep CFR~\cite{brown2019deep}. The difference lies in that $\Delta$ is a variable related to the belief model $\Theta$ and time $T$. We can prove that as time $T$ extends, its bound will always converge, meaning that a {Bayesian} Nash equilibrium can always be found under the new regret representation. For the detailed proof, please see the appendix.

%% file: 05experiment.tex
\section{Numerical Evaluation}
\label{evaluation}

To validate the effectiveness of Bayesian-CFR, we implement a number of baselines including CFR~\cite{zinkevich2007regret}, CFR+~\cite{tammelin2014solving}, MCCFR~\cite{lanctot2009monte}, Deep CFR~\cite{brown2019deep}, and DQN~\cite{mnih2013playing}, and compare these with Bayesian-CFR, Bayesian-CFR+, and deep Bayesian-CFR.
To mitigate the instability caused by the powerful fitting capabilities of neural networks, we use a simple deep neural network with limited parameters per layer.
Additionally, to assess the efficacy of the Bayesian-CFR technique, we conduct ablation studies by removing posterior belief updates and also consider an ideal game with complete information, serving  as the fundamental limit for Bayesian games. 

Performance can be measured in terms of exploitability, which is quantified in milli-big blinds per game (mbb/g), a standard metric for measuring win rates in Texas Hold'em.
The lower the exploitability, the closer it is to a Nash equilibrium (or Bayesian Nash equilibrium in this paper).
Experiments are performed on a server with an AMD EPYC 7513 32-Core Processor CPU and an NVIDIA RTX A6000 GPU.

\paragraph{Evaluation Environment: Texas hold'em.}
Our evaluations focus on Texas hold'em~\cite{southey2012bayes}, a game where each player receives two private cards and shares five community cards displayed in three stages—flop, turn, and river. The goal is to construct the best five-card hand possible. This setup allows for evaluating decision-making and strategic behavior, providing a controlled yet realistic setting for testing our algorithm against programmatically generated players.
We use the environment provided by ~\cite{zha2019rlcard}.

We consider four types of (latent) payoffs to create players with different behaviors, preference, and goals. 
This includes: {(i) Normal {payoff}} where the winners split the prize pool with some pre-determined ratio, and the losers lose all the stakes they have contributed;
{(ii) Conservative payoff}
where the winners receive a fixed reward, while the loser receives the negative of the reward; and
{(iii) aggressive payoff}
where winners receive a higher reward when the prize pool is larger and a lower reward when the prize pool is smaller, while the losers receive the opposite.
These latent payoffs are assigned to different players at the beginning of each game and represent some latent behavior, preference, and goals (denoted as a type), which can only be inferred by other players through observation/interaction history during the game.

In the following experiments, we categorize players into two classes over the different payoff types: pure-type players that have a single type (e.g., Pure-N, Pure-C, and Pure-A for normal, conservative, and aggressive payoffs) and mixed-type players whose type is some probability distribution. We use a total of nine probability distributions, denoted as Mixed-x, in the evaluation results.  
The specific distribution settings can be found in the appendix.

\begin{figure*}[h]
\centering
  \subfigure[Pure-N]{\includegraphics[width=0.33\textwidth]{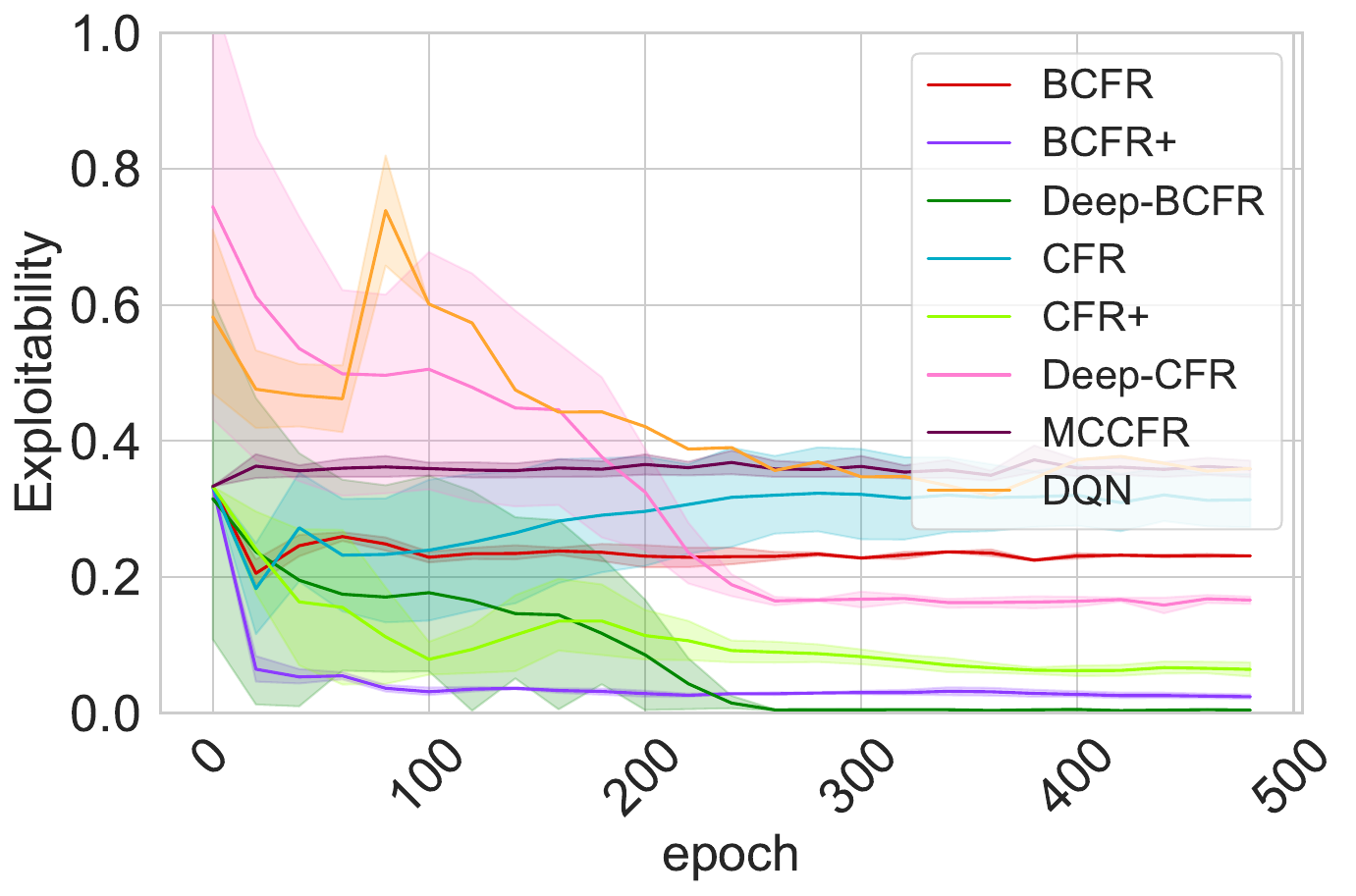}}\hfill
  \subfigure[Pure-C]{\includegraphics[width=0.33\textwidth]{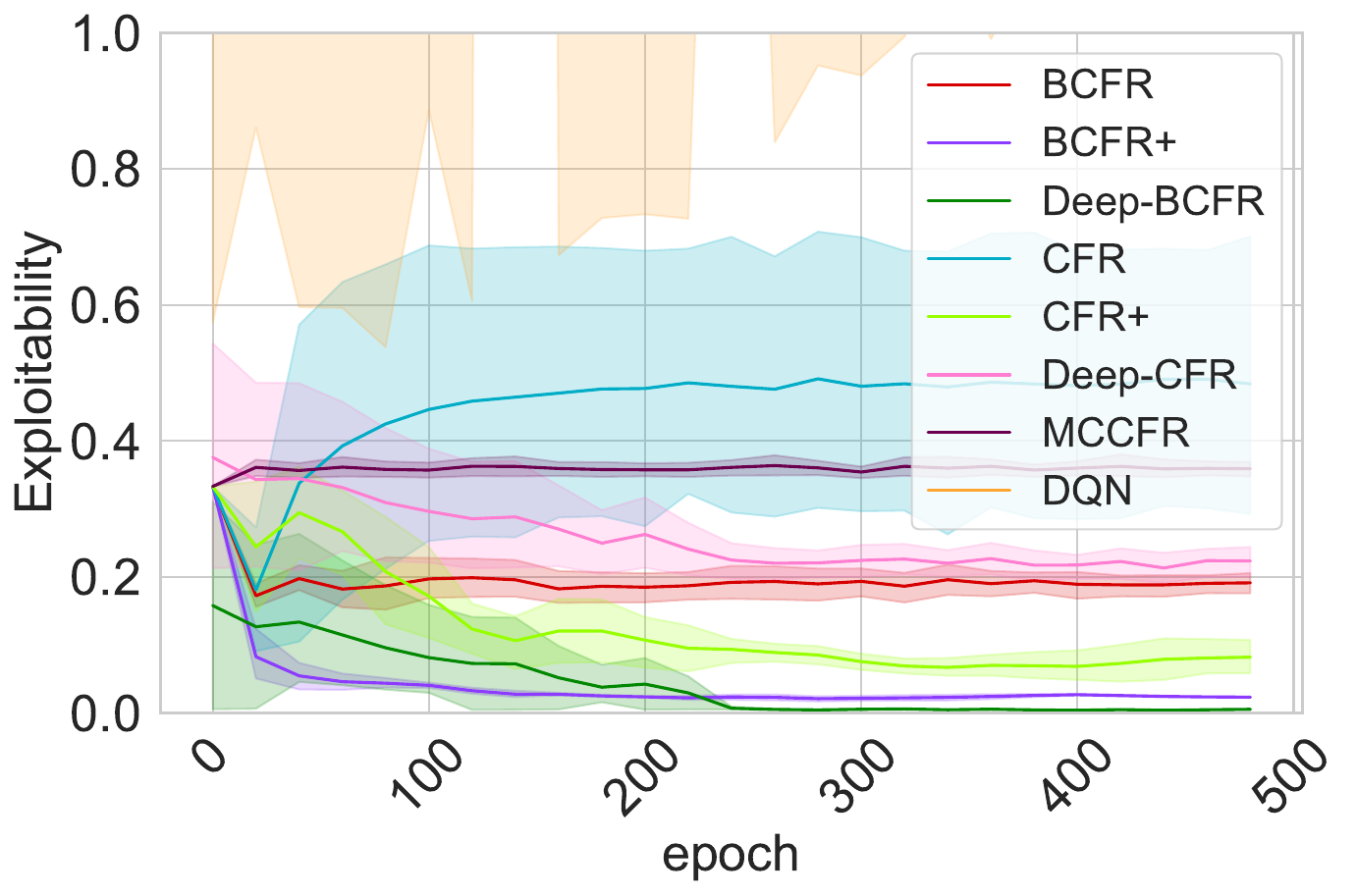}}\hfill
  \subfigure[Pure-A]{\includegraphics[width=0.33\textwidth]{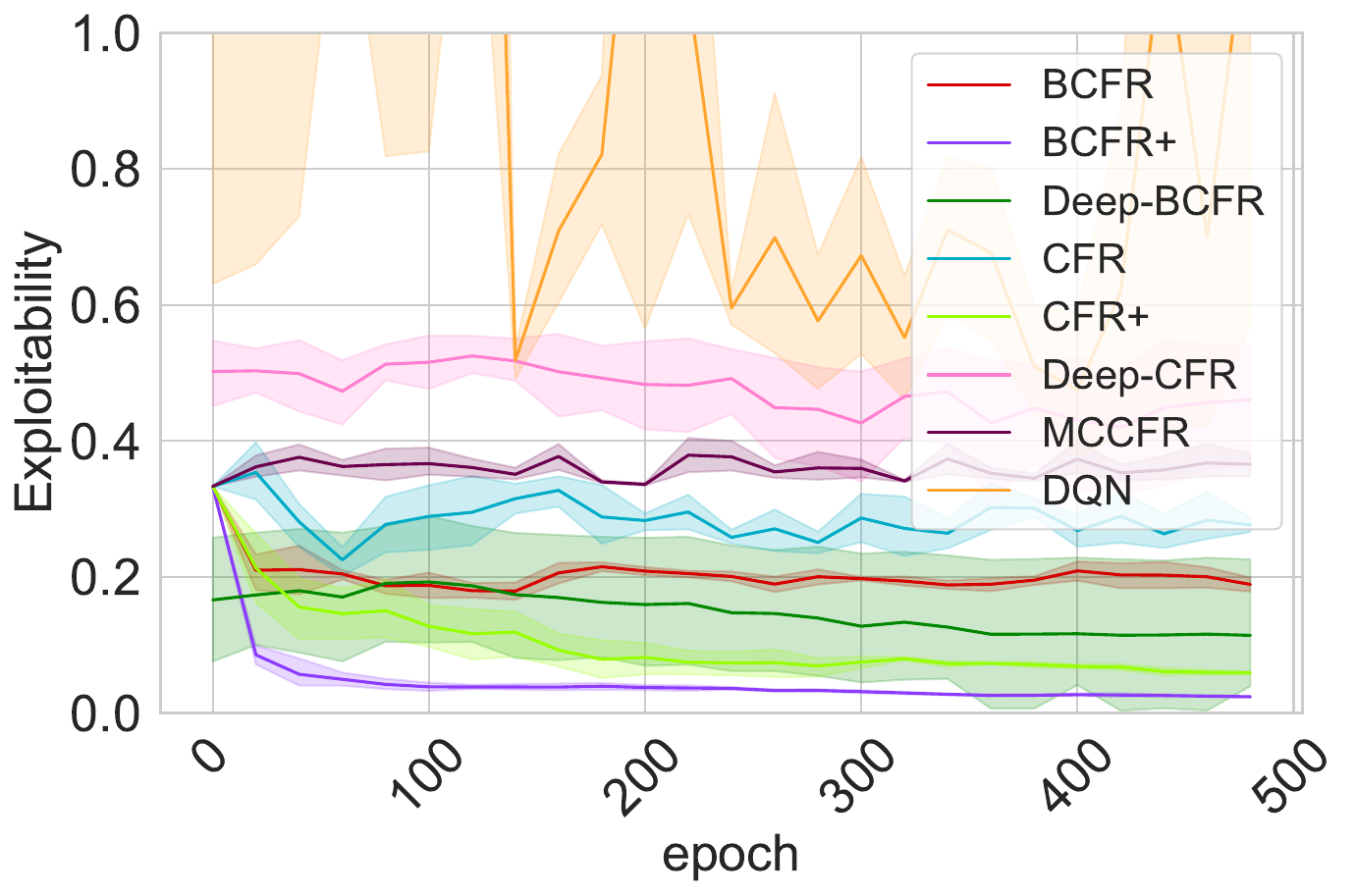}}\hfill
  \\
  \subfigure[Mixed-1]{\includegraphics[width=0.33\textwidth]{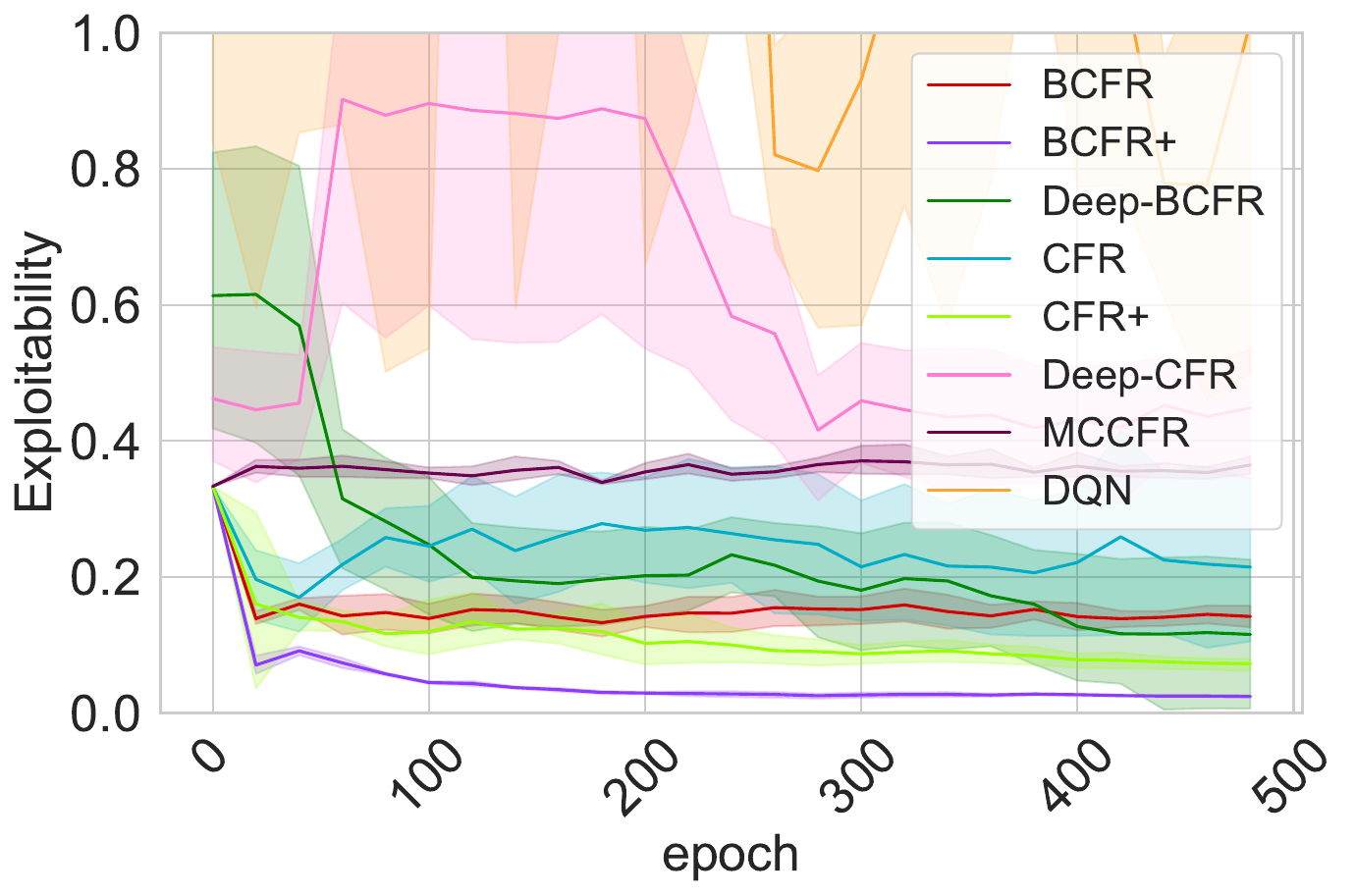}}\hfill
  \subfigure[Mixed-2]{\includegraphics[width=0.33\textwidth]{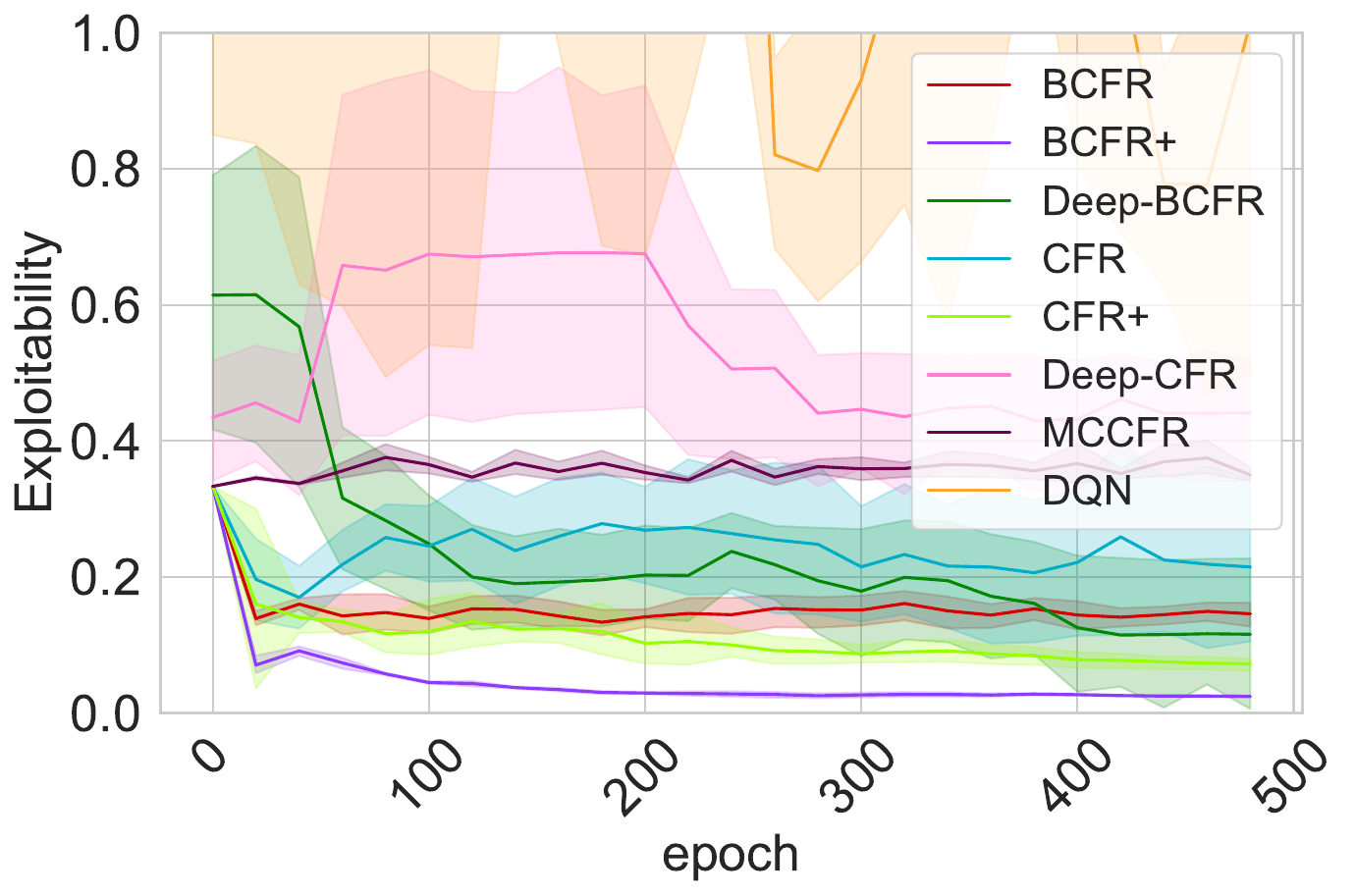}}\hfill
  \subfigure[Mixed-3]{\includegraphics[width=0.33\textwidth]{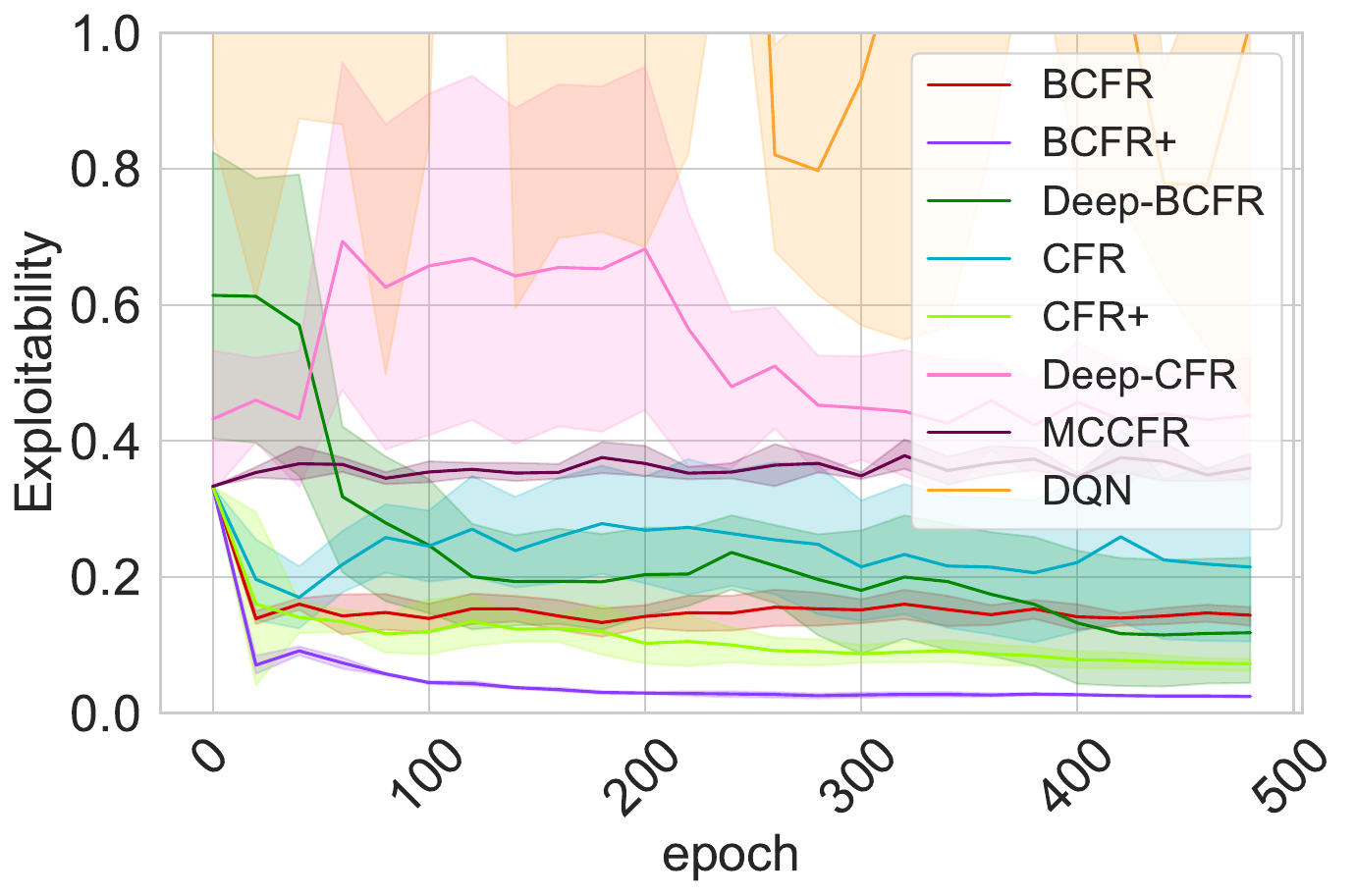}}\hfill
  \caption{ A comparison of our proposed Bayesian-CFR algorithms with baselines, including CFR~\cite{zinkevich2007regret}, CFR+~\cite{tammelin2014solving}, MCCFR~\cite{lanctot2009monte}, Deep CFR~\cite{brown2019deep}, and DQN~\cite{mnih2013playing} in Texas hold’em. The top row shows Bayesian games with pure-type players, and the bottom row shows Bayesian games with mixed-type players. A smaller exploitability implies closer ``distance" to the Nash Bayesian Equilibrium. Our Bayesian-CFR algorithms significantly outperforms all baselines.
  }\label{fig:Exp1}
\end{figure*}

\input{table0}

\paragraph{Evaluation against baselines.}
We evaluate the {\it exploitability} of the proposed algorithms, Bayesian-CFR,  Bayesian-CFR+, and Deep Bayesian-CFR and compare them against baselines, in Bayesian games with both pure-type and mixed-type players.
The results, as shown in Fig.~\ref{fig:Exp1}, demonstrate that the Bayesian-CFR algorithms can effectively converge and significantly improve performance (i.e., minimizing exploitability).
When facing players with unknown types/properties in Bayesian games, all CFR-based algorithms that utilize our proposed Bayesian regret achieve better performance than those that do not use Bayesian regret.
The performance of the reinforcement learning like the DQN algorithm is unstable and has difficult to converge. 

We also verify that for mixed-type players, the proposed Bayesian-CFR algorithms can achieve a similar exploitability and outperform baseline algorithms. Bayesian-CFR+ (BCFR+) achieves stable exploitability in all cases in our evolution. The exploitability values are summarized in Table~\ref{tab:sum} .
These results demonstrate that the Bayesian-CFR family of algorithms can effectively converge and significantly improve the performance in Bayesian games.
Plots of all evaluation results can be found in the appendix.

\textbf{Ablation studies.} 
To further study the contribution of the Bayesian module in our design, we conducted additional experiments by varying Bayesian-CFR (BCFR): BCFR without updating the posterior belief distribution and CFR that further removes the Bayesian beliefs entirely. We also compute the exploitability of an ideal game with complete information, which serves as a fundamental (non-achievable) lower bound for Bayesian games. Our numerical results are summarized in Table~\ref{tab:abl}.
In Bayesian games, Bayesian-CFR archives an exploitability very close to that of an ideal case of complete information games (0.19 vs. 0.16). This is mainly due to players' ability to reason about others' types through belief model updates, throughout the game.

\input{table1}

%% file: table0.tex
\begin{table*}[h!]

\centering
\resizebox{0.985\textwidth}{!}{%
\begin{tabular}{@{}clcccccccc@{}}
\toprule\toprule
\multirow{2}*{Type} & & \multicolumn{3}{c}{Ours} & \multirow{2}*{CFR} & \multirow{2}*{CFR+} & \multirow{2}*{Deep-CFR} & \multirow{2}*{MCCFR} & \multirow{2}*{DQN}\\ 
\cline{3-5}
 & & BCFR & BCFR+  & Deep-BCFR &  &  &  &  & \\ \midrule
Pure-N  &  & 0.22  & 0.02  & 0.01 &  0.33 & 0.06 & 0.15 & 0.35 & 0.59   \\
Pure-C  &  & 0.19  & 0.02  & 0.01 & 0.48 & 0.08 & 0.21 & 0.35 & 3.23   \\
Pure-A  &  & 0.19  & 0.02  & 0.11 & 0.22 & 0.05 & 0.45 & 0.37 & 1.90     \\ \midrule

Mixed-1  &  & 0.14  & 0.02  & 0.12  &  0.22 & 0.07 & 0.46 & 0.34 & 0.77   \\ 
Mixed-2  &  & 0.14  & 0.02  & 0.12  &  0.21 & 0.06 & 0.45 & 0.34 & 0.76    \\ 
Mixed-3  &  & 0.15  & 0.02  & 0.11  & 0.21 & 0.07 & 0.44 & 0.35 & 0.77 \\ \midrule
Averaged  &  & 0.17  & 0.02  & 0.08  &  0.28 & 0.07 & 0.36 & 0.35 &  1.34   \\
\bottomrule\bottomrule
\end{tabular}%
}
\caption{ This table summarizes the achieved exploitability for Bayesian-CFR, Bayesian-CFR+, and Deep Bayesian-CFR, as well as the baseline algorithms. We observe significant improvement over all baselines without using beilef models. In particular, Bayesian-CFR+ achieves the optimal result with an exploitability of 0.02 in all games. The results demonstrate that our proposed Bayesian-CFR algorithms can efficiently solve Bayesian games with both pure-type and mixed-type players. 
}
\label{tab:sum}
\vskip -0.2in
\end{table*}


%% file: table1.tex
\begin{table*}[h!]

\centering
\resizebox{0.55\textwidth}{!}{%
\begin{tabular}{@{}clcccc@{}}
\toprule\toprule
\multirow{2}*{Type} & & \multirow{2}*{CIG} & \multirow{2}*{BCFR}  & BCFR  & \multirow{2}*{CFR}\\
 & &  &   & w.o. posterior & \\
\midrule
Pure-N  &  & 0.18  & 0.22  & 0.31          &  0.33   \\
Pure-C  &  & 0.08  & 0.19  & 0.42 & 0.48   \\
Pure-A  &  & 0.16  & 0.19  & 0.22 & 0.22     \\ \midrule

Mixed-1  &  & 0.11  & 0.14  & 0.21  &  0.22   \\ 
Mixed-2  &  & 0.13  & 0.14  & 0.20  &  0.21    \\ 
Mixed-3  &  & 0.15  & 0.15  & 0.19  & 0.21  \\ \midrule
Mixed-4  &  & 0.19  & 0.22  & 0.31  &  0.32   \\ 
Mixed-5  &  & 0.20  & 0.23  & 0.31  &  0.33    \\ 
Mixed-6  &  & 0.22  & 0.23  & 0.30  & 0.33  \\ \midrule
Mixed-7  &  & 0.17  & 0.19  & 0.25  &  0.27   \\ 
Mixed-8  &  & 0.17  & 0.20  & 0.25  &  0.28    \\ 
Mixed-9  &  & 0.15  & 0.18  & 0.24  & 0.27  \\ \midrule
Averaged  &  & 0.16  & 0.19  & 0.27  &  0.31   \\
\bottomrule\bottomrule
\end{tabular}%
}
\caption{
Our ablation study compares Bayesian-CFR (BCFR) with a modified version without performing posterior belief updates during the game (i.e., BCFR w.o. posterior) and further removes the belief model all together (i.e., CFR). We also consider an ideal case of complete information games to obtain a fundemental lower bound for Bayesian games. Our results show that BCFR obtains exploitability close to the lower bound, while the improvements are mainly due to modeling of other players through belief model updates in Bayesian counterfactual regret minimization. 
}
\label{tab:abl}
\vskip -0.15in
\end{table*}


%% file: 06Conclusion.tex
\section{Conclusion}
In this paper, we present Bayesian-CFR, a novel algorithm for solving Bayesian Nash Equilibria in Bayesian games. 
We present a kernel-density method to estimate the posterior belief distribution to reason about the beliefs about the game and about other players’ types. This allows us to define the Bayesian regret and present a Bayesian-CFR minimization algorithm for computing the Bayesian Nash Equilibria. The Bayesian-CFR algorithm, together with its extensions to Bayesian-CFR+ and Deep Bayesian CFR, are shown to significantly outperform existing methods in classical Texas Hold’em games. In the overwhelming majority of real-life situations, players
have only partial information about the payoff relevant data of the game and the other players. Our work supports modeling autonomous agent interaction by reasoning about the behaviors and various properties of interest of other agents.

%% file: 07appendix.tex
\newpage
\appendix

\section{Proof of Bayesion}

\begin{lemma}\cite{mandyam2023kernel}
\label{lem:1}
Let $w_{m},w^{'}_{m} > 0$ be the bandwidths chosen to estimate the joint probability and marginal probability, respectively. with $mw_{m}^{d/2}\rightarrow \infty$ and $mw_{m}^{',d^{'}/2}\rightarrow \infty$ as $m\rightarrow\infty$, where $d$ is the dimension of $(h,\Theta)$ and $d^{'}$ is the dimension of $\Theta$. Then,
\begin{equation}
    \lim\limits_{m\rightarrow \infty}\widehat{Pr}_{\chi}(O_{\chi}|\Theta) = Pr_{\chi}(O_{\chi}|\Theta).
\end{equation}
\end{lemma}
\begin{proof} of Lemma~\ref{lem:1}
Due to the continuous nature of the likelihood function and the analysis in~\cite{wand1994kernel,chacon2018multivariate} about the bounded mean squared integral error in multivariate kernel density estimators, we know that as the sample size increases, the mean of the density estimator converges. Also, the boundedness of the mean squared integral ensures that the mean tends toward an unbiased estimator. Moreover, because our defined kernel function $K$ is square-integrable and twice differentiable, we know that its value will converge to a constant as the data size increases. Also, since this kernel function is an estimation of the true posterior, the following formula holds:
\begin{equation}
\begin{aligned}
\lim\limits_{m\rightarrow\infty}\frac{1}{m}\sum_{j=1}^{m}K(\frac{d_{s}(h_{i},h_{j})}{w})K^{'}(\frac{d_r{(\Theta,\Theta_j})}{w^{'}}) = Pr(O_{\chi},\Theta)
\end{aligned}
\end{equation}
The same as:
\begin{equation}
\begin{aligned}
\lim\limits_{m\rightarrow\infty}\frac{1}{m}\sum_{l=1}^{m}K(\frac{d_r{(\Theta,\Theta_l})}{w^{'}}) = Pr(\Theta)
\end{aligned}
\end{equation}
By the Continuous Mapping Theorem~\cite{mann1943stochastic}, we can get:
\begin{equation}
\begin{aligned}
\widehat{Pr}_{\chi}(O_{\chi}|\Theta) = 
\lim\limits_{m\rightarrow\infty}\frac{1}{m}\sum_{j=1}^{m}\frac{K(\frac{d_{s}(h,h_{j})}{w})K^{'}(\frac{d_r{(\Theta,\Theta_j})}{w^{'}})}{\sum_{l=1}^{m}K(\frac{d_r{(\Theta,\Theta_l})}{w^{'}})} = \frac{Pr(O_{\chi},\Theta)}{Pr(\Theta)} = Pr_{\chi}(O_{\chi}|\Theta)
\end{aligned}
\end{equation}
\end{proof}

\begin{proof} of Theorem~\ref{the:bayesian_converge}

According to assumption~\ref{ass:equal}. we can denote equivalence class $[\Theta] = \{\Theta^{'}:\Theta^{'}\simeq \Theta\}$, and the quotient space is defined as $\mathit{\widetilde{\Theta}} = \mathit{\Theta}/\simeq = \{[\Theta]:\Theta \in \mathit{\Theta}\}$. The corresponding canonical projection is denoted by $\pi:\mathit{\Theta}\rightarrow\mathit{\widetilde{\Theta}}$. Then, the projection $\pi$ induces a prior distribution on $\mathit{\widetilde{\Theta}}$ denoted by $\mathcal{\widetilde{P}}:\mathcal{\widetilde{P}}(A)=\mathcal{P}(\pi^{-1}(A))$.
Additionally, we can get:
\begin{equation}
||Pr(\cdot|\Theta_{1}) - Pr(\cdot|\Theta_{2})||_{L_{1}} = 0 \Leftrightarrow Pr(\cdot|\Theta_{1}) = Pr(\cdot|\Theta_{2}), a.e. \Leftrightarrow \Theta_{1}\simeq \Theta_{2} \Leftrightarrow [\Theta_{1}] = [\Theta_{1}]
\end{equation}
Given that $\text{KL}(\Theta,\Theta^{*}) = \text{KL}([\Theta],[\Theta^{*}])$ by the definition of the equivalence class.
Then, $A=\{[\Theta]:\text{KL}([\Theta],[\Theta^{*}])<\epsilon\}$, and $\pi^{-1}(A)=\{\Theta:\text{KL}(\Theta,\Theta^{*})<\epsilon\}$.
So $\mathcal{\widetilde{P}}(\{[\Theta]:\text{KL}([\Theta],[\Theta^{*}])<\epsilon\})=\mathcal{P}(\{\Theta:\text{KL}(\Theta,\Theta^{*})<\epsilon\})>0$ for any $\epsilon > 0$, that is, the $\text{KL}$ support condition is satisfied.

According to~\cite{van2000asymptotic,ghosal2017fundamentals,mandyam2023kernel}, the Bayesian model is parameterized by $[\Theta]$, and $\Theta$ is a compact set that exists for consistent tests. Then, according to~\cite{schwartz1965bayes}, the posterior $\widetilde{Pr}_{m}^{n}$ on $\Theta$ is consistent. So, for any $\epsilon > 0$, $\lim\limits_{m\rightarrow\infty\atop n\rightarrow\infty}Pr_{m}^{n}(\{\Theta:||Pr(\cdot|\Theta)-Pr(\cdot|\Theta^{*})||_{L_{1}}<\epsilon\}) = 1$.
\end{proof}

\section{Proof of Theorem~\ref{the:regret}}
\begin{proof}
Define $D(I)$ as the information sets of player $i$ that can be reached from $I$ (including $I$ itself).
Define $\sigma_{\theta}|{D(I)\rightarrow\sigma^{'}}$ as a strategy profile identical to $\sigma_{\theta}$ except within the information sets in $D(I)$, where it matches $\sigma_{\theta}^{'}$. The \textbf{Bayesian full counterfactual regret} at $I$ is given by:

\begin{equation}
  R_{i,\Theta,full}^{T}(I) = \frac{1}{T} \max\limits_{a\in \sum_i}\sum\limits_{t=1}^{T}\sum\limits_{\theta\in \Theta}Pr(\theta|O_{\chi}^{t})\pi_{-i,\theta}^{\sigma^{t}}(I)(u_{i,\theta}(\sigma_{\theta}^{t}|_{D(I)\rightarrow \sigma_{\theta}^{'}}, I)-u_{i,\theta}(\sigma_{\theta}^{t},I)).  
\end{equation}

Similarly, let $R_{i,\Theta,full}^{T,+}(I) = \max(R_{i,\Theta,full}^{T}(I),0)$.
Moreover, we define $\text{S}^{\sigma_{\theta}}_{i}(I, a)$ as the set of all probability at the next information sets of player $i$, where $a\in A(I)$ and $\sigma_{\theta}$ is the current strategy.
If $\sigma_{\theta}$ implies that $I$ is unreachable because of an action of player $i$, that action is changed to allow $I$ to be reachable.
 Define $\text{S}_{i}(I)=\bigcup_{a\in A(I)}\text{S}_{i}(I,a)$

\begin{lemma}
$R_{i,\Theta,full}^{T}(I)\leq R_{i,\Theta,imm}^{T}(I)+\sum_{I^{'}\in \text{S}_{i}(I)}R_{i,\Theta,full}^{T}(I)$
\end{lemma}
\begin{corollary}
\label{cor:bound}
$R_{i,\Theta,full}^{T,+}(I)\leq R_{i,\Theta,imm}^{T,+}(I)+\sum_{I^{'}\in \text{S}_{i}(I)}R_{i,\Theta,full}^{T,+}(I)$
\end{corollary}

\begin{proof}
\begin{equation}
\begin{aligned}
&R_{i,\Theta,full}^{T}(I) = \frac{1}{T} \max\limits_{a\in A(I)}\max\limits_{\sigma_{\theta}\in \sum_i}\sum\limits_{t=1}^{T}\sum\limits_{\theta\in \Theta}Pr(\theta|O_{\chi}^{t})\pi_{-i,\theta}^{\sigma^{t}}(I)\\
&(u_{i,\theta}(\sigma_{\theta}^{t}|_{I\rightarrow a},I)-u_{i,\theta}(\sigma_{\theta}^{t},I)+\sum\limits_{I^{'}\in \text{S}_i(I,a)}\text{S}_i^{\sigma_{\theta}}(I^{'}|I,a)(\sigma_{\theta}^{t}|_{D(I)\rightarrow \sigma_{\theta}^{'}},I^{'})-u_{i,\theta}(\sigma_{\theta}^{t},I^{'}))). 
\end{aligned}
\end{equation}
so, we split the following two parts as:
\begin{equation}
\begin{aligned}
&R_{i,\Theta,full}^{T}(I) \leq \frac{1}{T} \max\limits_{a\in A(I)}\max\limits_{\sigma_{\theta}\in \sum_i}\sum\limits_{t=1}^{T}\sum\limits_{\theta\in \Theta}Pr(\theta|O_{\chi}^{t})\pi_{-i,\theta}^{\sigma^{t}}(I)(u_{i,\theta}(\sigma_{\theta}^{t}|_{I\rightarrow a},I)-u_{i,\theta}(\sigma_{\theta}^{t},I))\\
&+\frac{1}{T} \max\limits_{a\in A(I)}\max\limits_{\sigma_{\theta}\in \sum_i}\sum\limits_{t=1}^{T}\sum\limits_{\theta\in \Theta}Pr(\theta|O_{\chi}^{t})\pi_{-i,\theta}^{\sigma^{t}}(I)\sum\limits_{I^{'}\in \text{S}_i(I,a)}\text{S}_i^{\sigma_{\theta}}(I^{'}|I,a)(\sigma_{\theta}^{t}|_{D(I)\rightarrow \sigma_{\theta}^{'}},I^{'})-u_{i,\theta}(\sigma_{\theta}^{t},I^{'})).  
\end{aligned}
\end{equation}
The first part of the expression on the right hand side is the immediate regret.
\begin{equation}
\begin{aligned}
&R_{i,\Theta,full}^{T}(I) \leq R_{i,\Theta,imm}^{T}(I)\\
&+\frac{1}{T}\max\limits_{a\in A(I)}\max\limits_{\sigma_{\theta}\in \sum_i}\sum\limits_{t=1}^{T}\sum\limits_{\theta\in \Theta}Pr(\theta|O_{\chi}^{t})\sum\limits_{I^{'}\in \text{S}_i(I,a)}\pi_{-i,\theta}^{\sigma^{t}}(I^{'})(\sigma_{\theta}^{t}|_{D(I)\rightarrow \sigma_{\theta}^{'}},I^{'})-u_{i,\theta}(\sigma_{\theta}^{t},I^{'})).  
\end{aligned}
\end{equation}

\begin{equation}
\begin{aligned}
&R_{i,\Theta,full}^{T}(I) \leq R_{i,\Theta,imm}^{T}(I)\\
&+\max\limits_{a\in A(I)}\sum\limits_{I^{'}\in \text{S}_i(I,a)}\frac{1}{T} \max\limits_{\sigma_{\theta}\in \sum_i}\sum\limits_{t=1}^{T}\sum\limits_{\theta\in \Theta}Pr(\theta|O_{\chi}^{t})\pi_{-i,\theta}^{\sigma^{t}}(I^{'})(\sigma_{\theta}^{t}|_{D(I)\rightarrow \sigma_{\theta}^{'}},I^{'})-u_{i,\theta}(\sigma_{\theta}^{t},I^{'})).  
\end{aligned}
\end{equation}
so,
\begin{equation}
\begin{aligned}
R_{i,\Theta,full}^{T}(I) \leq R_{i,\Theta,imm}^{T}(I)+\max\limits_{a\in A(I)}\sum\limits_{I^{'}\in \text{S}_i(I,a)}R_{i,\Theta,full}^{T}(I^{'})
\end{aligned}
\end{equation}
Because the game is perfect recall, given distinct $a, a^{'} \in A(I)$, $\text{S}_i(I, a)$ and $\text{S}_i(I, a^{’})$ are disjoint. Therefore:
\begin{equation}
\begin{aligned}
R_{i,\Theta,full}^{T,+}(I) \leq R_{i,\Theta,imm}^{T,+}(I)+\max\limits_{a\in A(I)}\sum\limits_{I^{'}\in \text{S}_i(I,a)}R_{i,\Theta,full}^{T,+}(I^{'})
\end{aligned}
\end{equation}
\end{proof}

\begin{lemma}
\label{lem:bound}
$R_{i,\Theta,full}^{T,+}(I)\leq \sum_{I^{'}\in D(I)} R_{i,\Theta,imm}^{T,+}(I^{'})$
\end{lemma}
\begin{proof}
We prove this for a particular game recursively on the size of $D(I)$.
If an information set has no S, then Corollary~\ref{cor:bound} proves the result.
We use this as a basic step. Also, observe that $D(I) = \{I\}\cup\bigcup_{I^{'}\in\text{S}_{i}(I)}(D(I^{'}))$, and that if $I^{'}\in\text{S}_{i}(I)$, then $I \notin D(I^{'})$, implying $|D(I^{'})|<|D(I)|$.
Thus, by induction, we can establish that:
\begin{equation}
\begin{aligned}
R_{i,\Theta,full}^{T,+}(I) \leq R_{i,\Theta,imm}^{T,+}(I)+\sum\limits_{I^{'}\in \text{S}_i(I,a)}\sum\limits_{I^{''}\in \text{S}_i(I,a)}R_{i,\Theta,imm}^{T,+}(I^{''})
\end{aligned}
\end{equation}
Because the game is perfect recall, for any distinct $I^{'},I^{''}\in\text{S}_{i}(I)$, $D(I^{'})$ and $D(I^{''})$ are disjoint. Therefore:
\begin{equation}
\begin{aligned}
R_{i,\Theta,imm}^{T,+}(I)+\sum\limits_{I^{'}\in \text{S}_i(I,a)}\sum\limits_{I^{''}\in \text{S}_i(I,a)}R_{i,\Theta,imm}^{T,+}(I^{''}) = \sum\limits_{I^{'}\in D(I)}R_{i,\Theta,imm}^{T,+}(I^{'})
\end{aligned}
\end{equation}
The result immediately follows.

If $P (\emptyset) = i$, then $R_{i,\Theta,full}^{T}(\emptyset) = R_{i,\Theta}^{T} $, and the theorem follows from Lemma~\ref{lem:bound}. If this is not the case, then we can simply add a new information set at the beginning of the game, where player $i$ only has one action. So theorem~\ref{the:regret} can be proved.
\end{proof}
\end{proof}

\section{Proof of Theorems~\ref{the:bound1}}
Regret matching can be defined in a domain with a fixed set of actions $A$ and a payoff $u_{\theta}^{t}: A\rightarrow \mathbb{R}$ function. At each iteration $t$, a distribution over the actions, $\sigma_{\theta}^{t}$, is chosen based on the cumulative regret $R^{t}_{\Theta}$

\begin{equation}
    R_{\Theta}^{T}(a) = \frac{1}{T}\sum\limits_{t=1}^{T}\sum\limits_{\theta\in \Theta}Pr(\theta|O_{\chi}^{t})\Bigg( u_{\theta}^{t}(a) - \sum\limits_{a^{'}\in A}\sigma_{\theta}^{t}(a^{'})u_{\theta}^{t}(a^{'}) \Bigg)
\end{equation}
and define$ R_{\Theta}^{T,+}(a) = \max(R_{\Theta}^{T,+}(a), 0)$. To apply regret matching, one chooses the distribution:

\begin{equation}
\sigma_{\Theta}^{t}(a) = 
\left\{
\begin{aligned}
    &\frac{R_{+,\Theta}^{t-1}(a)}{\sum_{b\in A}R_{+,\Theta}^{t-1}(b)} & if \sum_{b\in A} R_{+,\Theta}^{t-1}(b) > 0\\
    &\frac{1}{\lvert A \rvert} & otherwise
    \end{aligned}
\right.
\end{equation}

\begin{lemma}
\label{lem:bound_R}
$\max_{a\in A}R^{T}_{\Theta}(a) \leq \frac{\Delta_{u,\Theta}\sqrt{|A|}}{\sqrt{T}}$, where $\Delta_{u,\Theta} = \max_{t\in\{1,..,T\}}\max_{a,a^{'}\in A}\sum\limits_{\theta \in \Theta}Pr(\theta|O_{\chi}^{t})\Big(u_{\theta}^{t}(a)-u_{\theta}^{t}(a^{'})\Big)$
\end{lemma}
\begin{proof}
According to Theorem~\ref{the:bayesian_converge}, we can know $\sum_{\theta\in \Theta}Pr(\theta|O_{\chi}^{T})R_{\theta}^{T}\simeq R_{\Theta}^{T}$, and according to~\cite{zinkevich2007regret},we can know $\max_{a\in A}R^{T}_{\theta}(a) \leq \frac{\Delta_{u,\theta}\sqrt{|A|}}{\sqrt{T}}$. So
\begin{equation}
    R_{\Theta}^{T} \leq \sum_{\theta\in \Theta}Pr(\theta|O_{\chi}^{T})\frac{\Delta_{u,\theta}\sqrt{|A|}}{\sqrt{T}}.
\end{equation}
Similarly, we can get $\max_{a\in A}R^{T}_{\Theta}(a) \leq \frac{\Delta_{u,\Theta}\sqrt{|A|}}{\sqrt{T}}$.

\textbf{Proof Theorems~\ref{the:bound1}}
According to the Eqn.~\ref{eqn:regret_imm} we can get for all $I\in \mathcal{I}_{i},a\in A(I),\sum_{\theta\in \Theta}Pr(\theta|O_{\chi}^{t})\pi_{-i,\theta}^{\sigma_{\Theta}^{t}}(u_{i,\theta}(\sigma_{\theta}^{t}|_{I\rightarrow a})-u_{i,\theta}(\sigma_{\theta}^{t},I))\leq \Delta_{u,i,\Theta}$.

And according to Lemma~\ref{lem:bound_R}, we can get that the counterfactual regret of that node will be less than $\frac{\Delta_{u,i,\Theta}\sqrt{|A(I)|}}{\sqrt{T}}\leq \frac{\Delta_{u,i,\Theta}\sqrt{|A|}}{\sqrt{T}}$
\end{proof}

\section{Proof of Theorem~\ref{the:deep_bound}}
\begin{proof}
    
Assuming an online learning scheme:
\begin{equation}
\label{eql:sigma}
\sigma^{t}_{\theta}(I,a)= \left\{
\begin{aligned}
\frac{y_{+,\theta}^{t}(I,a)}{\sum_{a}y_{+,\theta}^{t}(I,a)} & , & \text{if} \sum_{a} y_{+,\theta}^{t}(I,a) > 0 \\
\text{arbitrary} & , & \text{otherwise}
\end{aligned}
\right.
\end{equation}
Corollary 3.0.6 in~\cite{morrill2016using}provides the upper bound of the total regret by leveraging a function of the L2 distance between $y_{t,\theta}^{+}$ and $R_{\theta}^{T,+}$ on each Infoset: $R_{\theta}^{T,+}$

\begin{equation}
    \begin{aligned}
        \max\limits_{a\in A}(R_{\theta}^{T}(I,a))^{2} &\leq |A|\Delta_{\theta}^{2}T + 4 \Delta_{\theta} |A| \sum\limits_{t=1}^{T}\sum\limits_{a\in A} \sqrt{(R_{+,\theta}^{t}(I,a) - y_{+,\theta}^{t}(I,a))^{2}}\\
        &\leq |A|\Delta_{\theta}^{2}T + 4\Delta_{\theta}|A|\sum\limits_{t=1}^{T}\sum\limits_{a\in A} \sqrt{(R_{\theta}^{t}(I,a) - y_{\theta}^{t}(I,a))^{2}}.
    \end{aligned}
\end{equation}

As shown in Eqn~\ref{eql:sigma},$\sigma^{t}(I,a)$ is invariant to rescaling across all actions at an infoset. Also, for any $C(I)>0$
\begin{equation}
    \begin{aligned}
        \max\limits_{a\in A}(R_{\theta}^{T}(I,a))^{2} &\leq |A|\Delta_{\theta}^{2}T + 4 \Delta_{\theta} |A| \sum\limits_{t=1}^{T}\sum\limits_{a\in A} \sqrt{(R_{\theta}^{t}(I,a) -C(I) y_{\theta}^{t}(I,a))^{2}}.
    \end{aligned}
\end{equation}

Let $x^{t}(I)$ be an indicator variable which is $1$ if $I$ was traversed on iteration $t$. If $I$ was traversed then $\widetilde{r}^{t}(I)$ was stored in $M_{V,p}$, otherwise $\widetilde{r}^{t}(I) = 0$. Assuming, for now, that $M_{V,p}$ is not full, so all sampled regrets are stored in the memory.

Let $\Pi^{t}$ be the fraction of iterations on which $x^{t}(I) = 1$, and

\begin{equation}
    \begin{aligned}
        \epsilon^{t}(I) = \left| \left| E_{t}[\widetilde{r}^{t}(I)|x^{t}(I) = 1] - V(I,a|\theta^{t}) \right| \right|_{2}.
    \end{aligned}
\end{equation}
By inserting canceling factor of $\sum_{t^{'} =1}^{t}x^{t^{'}}(I)$ and setting $C(I) = \sum_{t^{'} =1}^{t}x^{t^{'}}(I)$, we obtain
\begin{equation}
    \begin{aligned}
        \max\limits_{a\in A}(R_{\theta}^{T}(I,a))^{2} &\leq |A|\Delta_{\theta}^{2}T + 4 \Delta_{\theta} |A| \sum\limits_{t=1}^{T}(\sum\limits_{t^{'}=1}^{t}x^{t^{'}}(I))\sum\limits_{a\in A} \sqrt{\frac{\widetilde{R}_{\theta}^{t}(I,a)}{\sum_{t^{'}=1}^{t}x^{t^{'}}(I)} - y_{\theta}^{t}(I,a))^{2}}\\
        & = |A|\Delta_{\theta}^{2}T + 4 \Delta_{\theta} |A| \sum\limits_{t=1}^{T}(\sum\limits_{t^{'}=1}^{t}x^{t^{'}}(I)) \left| \left| E_{t}[\widetilde{r}_{\theta}^{t}(I)|x^{t}(I) = 1] - V(I,a|\theta_{\theta}^{t}) \right| \right|_{2}\\
        & = |A|\Delta_{\theta}^{2}T + 4 \Delta_{\theta} |A| \sum\limits_{t=1}^{T} t \Pi^{t}(I) \epsilon^{t}(I)\\
        & \leq |A|\Delta_{\theta}^{2}T + 4 \Delta_{\theta} |A| \sum\limits_{t=1}^{T} \Pi^{t}(I) \epsilon^{t}(I)
    \end{aligned}
\end{equation}

The first term of this expression is the same as the regret bound of tabular CFR algorithm, while the second term accounts for the approximation error. Theorem 3 in~\cite{brown2019deep} shows the regret bound for $K$-external sampling, where we can get the same results for the case of $K$-probe sampling. Thus, we can get

\begin{equation}\label{eq:39}
    \begin{aligned}
        \max\limits_{a\in A}(R_{\theta}^{T}(I,a))^{2} & \leq |A|\Delta_{\theta}^{2}TK^{2} + 4 \Delta_{\theta} |A| \sum\limits_{t=1}^{T}K^{2} \Pi^{t}(I) \epsilon^{t}(I).
    \end{aligned}
\end{equation}

The new regret bound in Eq.~\ref{eq:39} can be plugged into Theorem 3 in~\cite{lanctot2009monte}, similar to Theorem 4, to obtain

\begin{equation}
    \begin{aligned}
        \bar{R}_{p,\theta}^{T} \leq \sum\limits_{I\in \mathcal{I}_p}((1+\frac{\sqrt{2}}{\sqrt{\rho K}}) \Delta_{\theta} \frac{\sqrt{|A|}}{\sqrt{T}} + \frac{4}{\sqrt{T}}\sqrt{|A|\Delta_{\theta}\sum\limits_{t=1}^{T}\Pi^{t}(I)\epsilon^{t}(I)})
    \end{aligned}
\end{equation}

Simplifying the first term and rearranging,

\begin{equation}
    \begin{aligned}
        \bar{R}_{p,\theta}^{T} &\leq \sum\limits_{I\in \mathcal{I}_p}((1+\frac{\sqrt{2}}{\sqrt{\rho K}}) \Delta_{\theta} \frac{\sqrt{|A|}}{\sqrt{T}} + \frac{4\sqrt{|A|\Delta}}{\sqrt{T}}\sqrt{\sum\limits_{t=1}^{T}\Pi^{t}(I)\epsilon^{t}(I)})\\
        & =\sum\limits_{I\in \mathcal{I}_p}((1+\frac{\sqrt{2}}{\sqrt{\rho K}}) \Delta_{\theta} \frac{\sqrt{|A|}}{\sqrt{T}} + \frac{4\sqrt{|A|\Delta_{\theta}}}{\sqrt{T}}|\mathcal{I}_{p}|\frac{\sum_{I\in \mathcal{I}_{p}}}{|\mathcal{I}_{p}|}\sqrt{\sum\limits_{t=1}^{T}\Pi^{t}(I)\epsilon^{t}(I)})\\
        & \leq \sum\limits_{I\in \mathcal{I}_p}((1+\frac{\sqrt{2}}{\sqrt{\rho K}}) \Delta_{\theta} \frac{\sqrt{|A|}}{\sqrt{T}} + \frac{4\sqrt{|A|\Delta_{\theta}\mathcal{I}_{p}}}{\sqrt{T}}\sqrt{\sum\limits_{t=1}^{T}\sum_{I\in \mathcal{I}_{p}}\Pi^{t}(I)\epsilon^{t}(I)})
    \end{aligned}
\end{equation}

Now, lets consider the average MSE loss $\mathcal{L}_{V}^{T}(\mathcal{M}^{T})$ at time $T$ over the samples in memory $\mathcal{M}^{T}$

We start by stating two well-known lemmas:

\begin{lemma}
\label{lem:mse1}
The MSE can be decomposed into bias and variance components
\begin{equation}\label{eq:42}
    \mathbb{E}_{x}[(x-\theta)^{2}] = (\theta - \mathbb{E}[x])^{2} + Var(\theta)
\end{equation}
\end{lemma}

\begin{lemma}
\label{lem:mse2}
The mean of a random variable minimizes the MSE loss
\begin{equation}
    \arg\min\limits_{\theta}\mathbb{E}_{x}[(x-\theta)^{2}]=\mathbb{E}[x]
\end{equation}
and the value of the loss at $\theta = \mathbb{E}[x]$ is $Var(x)$.
\end{lemma}

\begin{equation}
\begin{aligned}
    \mathcal{L}_{V}^{T} & = \frac{1}{\sum_{I\in \mathcal{I}_{p}}\sum_{t=1}^{T}x^{t}(I)}x^{t}(I) \left | \left | \widetilde{r}^{t}(I) - V(I|\theta^{T}) \right | \right |_{2}^{2}\\
    & \geq \frac{1}{|\mathcal{I}_{p}|T} \sum\limits_{I\in \mathcal{I}_{p}} \sum\limits_{t=1}^{T}x^{t}(I) \left | \left | \widetilde{r}^{t}(I) - V(I|\theta^{T}) \right | \right |_{2}^{2}\\
    & = \frac{1}{|\mathcal{I}_{p}|}\sum\limits_{I\in \mathcal{I}_{p}}\Pi^{T}(I)\mathbb{E}[ \left | \left | \widetilde{r}^{t}(I) - V(I|\theta^{T}) \right | \right |_{2}^{2}|x^{t}(I) = 1]
\end{aligned}
\end{equation}

Let $V^{*}$ be the model that minimizes $\mathcal{L}^{T}$ on $\mathcal{M}_{T}$. Useing lemmas~\ref{lem:mse1} and~\ref{lem:mse2},

\begin{equation}
\begin{aligned}
    \mathcal{L}_{V}^{T} & \geq \frac{1}{|\mathcal{I}_{p}|}\sum\limits_{I\in \mathcal{I}_{p}}\Pi^{T}(I)( \left | \left | V(I|\theta^{T}) - \mathbb{E}[\widetilde{r}^{t}(I)|x^{t}(I) = 1)]\right | \right |_{2}^{2} + \mathcal{L}_{V^{*}}^{T})
\end{aligned}
\end{equation}

So,
\begin{equation}
\begin{aligned}
    \mathcal{L}_{V}^{T} - \mathcal{L}_{V^{*}}^{T}\geq \frac{1}{|\mathcal{I}_{p}|} \sum\limits_{I\in \mathcal{I}_{p}}\Pi^{T}(I)\epsilon^{T}(I)\\
    \sum\limits_{I\in \mathcal{I}_{p}}\Pi^{T}(I)\epsilon^{T}(I) \leq |\mathcal{I}_{p}|(\mathcal{L}_{V}^{T} - \mathcal{L}_{V^{*}}^{T})
\end{aligned}
\end{equation}

Plugging this into Eq.~\ref{eq:42}, we arrive at

\begin{equation}
    \begin{aligned}
        \bar{R}_{p,\theta}^{T} 
        & \leq (1+\frac{\sqrt{2}}{\sqrt{\rho K}}) \Delta_{\theta} |\mathcal{I}_{p}|\frac{\sqrt{|A|}}{\sqrt{T}} + \frac{4\sqrt{|A|\Delta_{\theta}\mathcal{I}_{p}}}{\sqrt{T}}\sqrt{|\mathcal{I}_{p}|\sum\limits_{t=1}^{T}(\mathcal{L}_{V}^{T} - \mathcal{L}_{V^{*}}^{T})}\\
        & \leq (1+\frac{\sqrt{2}}{\sqrt{\rho K}}) \Delta_{\theta} |\mathcal{I}_{p}|\frac{\sqrt{|A|}}{\sqrt{T}} + 4|\mathcal{I}_{p}|\sqrt{|A|\Delta_{\theta} \epsilon_{\mathcal{L}}}\\
    \end{aligned}
\end{equation}

\begin{equation}
    \begin{aligned}
        \sum\limits_{\theta\in \Theta}Pr(\theta|O_{\chi}^{T})\bar{R}_{p,\theta}^{T} 
        & \leq (1+\frac{\sqrt{2}}{\sqrt{\rho K}}) \sum\limits_{\theta\in \Theta}Pr(\theta|O_{\chi}^{T})\Delta_{\theta} |\mathcal{I}_{p}|\frac{\sqrt{|A|}}{\sqrt{T}}\\
        &+ \frac{4\sum\limits_{\theta\in \Theta}Pr(\theta|O_{\chi}^{T})\sqrt{|A|\Delta_{\theta}\mathcal{I}_{p}}}{\sqrt{T}}\sqrt{|\mathcal{I}_{p}|\sum\limits_{t=1}^{T}(\mathcal{L}_{V}^{T} - \mathcal{L}_{V^{*}}^{T})}\\
        & \leq (1+\frac{\sqrt{2}}{\sqrt{\rho K}})\sum\limits_{\theta\in \Theta}Pr(\theta|O_{\chi}^{T}) \Delta_{\theta} |\mathcal{I}_{p}|\frac{\sqrt{|A|}}{\sqrt{T}} + 4\sum\limits_{\theta\in \Theta}Pr(\theta|O_{\chi}^{T})|\mathcal{I}_{p}|\sqrt{|A|\Delta_{\theta} \epsilon_{\mathcal{L}}}\\
        & \leq (1+\frac{\sqrt{2}}{\sqrt{\rho K}})\sum\limits_{\theta\in \Theta}Pr(\theta|O_{\chi}^{T}) \Delta_{\theta} |\mathcal{I}_{p}|\frac{\sqrt{|A|}}{\sqrt{T}} + 4\sum\limits_{\theta\in \Theta}|\mathcal{I}_{p}|\sqrt{|A|Pr^{2}(\theta|O_{\chi}^{T})\Delta_{\theta} \epsilon_{\mathcal{L}}}\\
        & \leq (1+\frac{\sqrt{2}}{\sqrt{\rho K}})\sum\limits_{\theta\in \Theta}Pr(\theta|O_{\chi}^{T}) \Delta_{\theta} |\mathcal{I}_{p}|\frac{\sqrt{|A|}}{\sqrt{T}} + 4\sum\limits_{\theta\in \Theta}|\mathcal{I}_{p}|\sqrt{|A|Pr(\theta|O_{\chi}^{T})\Delta_{\theta} \epsilon_{\mathcal{L}}}\\
        & = (1+\frac{\sqrt{2}}{\sqrt{\rho K}})\Delta_{\Theta} |\mathcal{I}_{p}|\frac{\sqrt{|A|}}{\sqrt{T}} + 4|\mathcal{I}_{p}|\sqrt{|A|\Delta_{\Theta} \epsilon_{\mathcal{L}}}
    \end{aligned}
\end{equation}
According to Theorem~\ref{the:bayesian_converge}, we know that as time $T$ increases, $\sum\limits_{\theta\in \Theta}Pr(\theta|O_{\chi}^{T})\bar{R}_{p,\theta}^{T}$
  gradually converges. Therefore, $\bar{R}_{p,\Theta}^{T}\simeq\sum\limits_{\theta\in \Theta}Pr(\theta|O_{\chi}^{T})\bar{R}_{p,\theta}^{T} $

So, we can proof theorem~\ref{the:deep_bound}

So far we have assumed that $\mathcal{M}_{V}$ contains all sampled regrets. The number of samples in the memory at iteration $t$ is bounded by $K \cdot |\mathcal{I}_{p}| \cdot t$. Therefore, if$K \cdot |\mathcal{I}_{p}| \cdot T < |\mathcal{M}_{V}|$ then the memory will never be full, and we can make this assumption.

Let $\rho = T^{-\frac{1}{4}}$

\begin{equation}
    P(\bar{R}_{p}^{T} > (1+\frac{\sqrt{2}}{\sqrt{K}}) \Delta |\mathcal{I}_{p}|\frac{\sqrt{|A|}}{T^{-\frac{1}{4}}} + 4|\mathcal{I}_{p}|\sqrt{|A|\Delta \epsilon_{\mathcal{L}}})< T^{-\frac{1}{4}}
\end{equation}

Therefore, for any $\epsilon > 0$,

\begin{equation}
\lim_{T\rightarrow \infty} P(\bar{R}_{p}^{T}-4|\mathcal{I}_{p}|\sqrt{|A|\Delta \epsilon_{\mathcal{L}}}>\epsilon) = 0.
\end{equation}

\section{Pre-code}

\begin{algorithm}
    \caption{Bayesian-CFR (BCFR)}
    \label{alg:BCFR}
\begin{algorithmic}
    \STATE  Initialize cumulative regret  $R(I,a)$ so that it returns $0$ for all inputs for player $p $
    \STATE Initialize strategy  $\sigma(I,a)$ so that it returns $0$ for all inputs for player $p$
    \STATE Initialize type memories $Q_{\theta}$, which is a queue with $m$ length. And Initialize competitor-type memories $Q$, a queue with $n$ length.
    \FOR{CFR Iteration $t=1$ {\bfseries to} $T$}
    \STATE Sample type $\theta$ form the posterior $Pr_{t}(\Theta) $
    \FOR{\textbf{each} player $p$}
    \STATE BCFRTRVERSE$(\emptyset,p,t,1,1,\theta)$ 
    \ENDFOR
    \STATE Evaluate with the competitor and collect data to $Q$
    \STATE Calculate the likelihood $\widehat{Pr}_{m}^{n}$ of $\theta$ using Eqn.~\ref{eqn:likelihood} with training data $Q_{\theta}$ and competitor data $Q$.
    \STATE Update the posterior using Eqn.~\ref{eqn:bayesion}
    \ENDFOR
\end{algorithmic}
\end{algorithm}

\begin{algorithm}
    \caption{BCFRTRVERSE}
    \label{alg:trverse}
\begin{algorithmic}
    \STATE  \textbf{Function:} BCFRTRVERSE$(h,p,t,\pi_{p},\pi_{-p},\theta)$
    \STATE $Q_{\theta} \leftarrow Q_{\theta}\cup \{h\}$
    \IF{$h\in Z$}
    \STATE \textbf{return} $u_i(h)$
    \ELSIF{$h$ is a chance node}
    \STATE Sample an action $a$ from the probability $\sigma_{c}(h)$
    \STATE \textbf{return} BCFRTRVERSE$(h\cdot a,p,t,\pi_{p},\pi_{-p},\theta)$
    \ELSIF{$P(h)=p$}
    \STATE $I \leftarrow$ Information set containing $h$;
    \STATE $u_{\sigma}\leftarrow 0$
    \STATE $u_{\sigma_{I\rightarrow a}}(a)\leftarrow 0$ for all $a\in A(I)$
    \FOR{$a\in A(I)$}
    \IF{$P(h)=p$}
    \STATE $u_{\theta,\sigma_{I\rightarrow a}}(a) \leftarrow$ BCFRTRVERSE$(h\cdot a,p,\sigma^{t}(I,a)\cdot\pi_{p},\pi_{(-p)},\theta)$
    \ELSE
    \STATE $u_{\theta,\sigma_{I\rightarrow a}}(a) \leftarrow$ BCFRTRVERSE$(h\cdot a,p,\pi_{p},\sigma^{t}(I,a)\cdot\pi_{(-p)},\theta)$
    \ENDIF
    \STATE $u_{\sigma,\theta}\leftarrow u_{\sigma,\theta}+\sigma^{t}(I,a)\cdot u_{\theta,\sigma_{I\rightarrow a}}(a)$
    \ENDFOR
    \IF{$P(h) = p$}
    \FOR{$a\in A(I)$}
    \STATE $r(I,a)\leftarrow r(I,a)+\Pr(\theta|h)\cdot\pi_{-p}\cdot(v_{\sigma_{I\rightarrow a}}(a)-v_{\sigma})$
    \STATE $s(I,a)\leftarrow s(I,a)+\pi_{-p}\cdot\sigma^{t}(I,a)$
    \ENDFOR
    \STATE $\sigma^{t+1}(I)\leftarrow$ regret-matching values computed using Eqn~\ref{eqn:strategy} and regret table $r_I$
    \ENDIF
    \STATE \textbf{return} $u_{\sigma^{t}}$
    \ENDIF
\end{algorithmic}
\end{algorithm}

\begin{algorithm}
    \caption{Deep BCFR}
    \label{alg:BCFR}
\begin{algorithmic}
    \STATE  Initialize cumulative regret network $R(I,a|\psi_{p,\theta})$with parameters $\psi_{p,\theta}$ so that it returns 0 for all inputs.
    \STATE Initialize regret memories $\mathcal{M}_{r,1},\mathcal{M}_{r,2}$ and strategy memory $\mathcal{M}_{\pi,\Theta}$ 
    \STATE Initialize type memories $Q_{\theta}$, which is a queue with $m$ length. And Initialize competitor-type memories $Q$, a queue with $n$ length.
    \FOR{CFR Iteration $t=1$ {\bfseries to} $T$}
    \STATE Sample type $\theta$ form $Pr_{t}(\Theta) $
    \FOR{\textbf{each} player $p$}
    \FOR{traverse $k=1$ {\bfseries to} $K$}
    \STATE DEEPBCFRTRVERSE$(\emptyset,R,\psi_{1,\theta},\psi_{2,\theta},\mathcal{M}_{r,p},\mathcal{M}_{\pi,\Theta},\theta)$
    \ENDFOR
    \STATE train $\psi_{p,\theta}$ on loss for player $p$ $\mathcal{L} = \mathbb{E}_{(I,\widetilde{r})\sim M_{r,p}}[\sum_{a}((R(\cdot|\psi^{t}_{p,\theta})+\widetilde{r})^{+} - R(\cdot|\psi^{t+1}_{p,\theta}))^{2}]$
    \ENDFOR
    \STATE Evaluate with competitor and collect data to $Q$
    \STATE Calculate the likelihood $\widehat{Pr}_{m}^{n}$ of $\theta$ using Eqn.~\ref{eqn:likelihood} with training data $Q_{\theta}$ and competitor data $Q$.
    \STATE Update the posterior using Eqn.~\ref{eqn:bayesion}
    \ENDFOR
    \STATE Train $\psi_{\pi,p}$ on loss for player $p$
    \STATE $\mathcal{L} = \mathbb{E}_{(I,\widetilde{\pi})\sim M_{\pi,p}}[\sum_{a}((S(\cdot|\psi^{t}_{p})+\widetilde{\pi})^{+} - S(\cdot|\psi^{t+1}_{p}))^{2}]$
    \STATE \textbf{return} $\psi_{\pi,V},\psi_{\pi,T}$
\end{algorithmic}
\end{algorithm}

\begin{algorithm}
    \caption{DEEPBCFRTRVERSE}
    \label{alg:trverse}
\begin{algorithmic}
    \STATE  \textbf{Function:} DEEPBCFRTRVERSE$(h,p,\psi_{p},\psi_{-p},\mathcal{M}_{r,p},\mathcal{M}_{\pi,(-p)},\theta)$
    \STATE $Q_{\theta} \leftarrow Q_{\theta}\cup \{h\}$
    \IF{$h\in Z$}
    \STATE \textbf{return} $u_i(h)$
    \ELSIF{$h$ is a chance node}
    \STATE Sample an action $a$ from the probability $\sigma_{c}(h)$
    \STATE \textbf{return} DEEPBCFRTRVERSE$(h\cdot a,p,\psi_{p},\psi_{-p},\mathcal{M}_{r,p},\mathcal{M}_{\pi,(-p)},\theta)$
    \ELSIF{$P(h)=p$}
    \STATE $I \leftarrow$ Information set containing $h$;
    \STATE $\sigma^{t}(I) \leftarrow$ Strategy of Information set $I$ computed from $R(I,a|\psi_{p})$ using regret matching ;
    \FOR{$a\in A(I)$}
    \STATE $u_{\theta}(a) \leftarrow$ DEEPBCFRTRVERSE$(h\cdot a,p,\psi_{p},\psi_{-p},\mathcal{M}_{r,p},\mathcal{M}_{\pi,(-p)},\theta)$
    \ENDFOR
    \STATE $u_{\sigma^{t},\theta} \leftarrow \sum_{a\in A(I)}\sigma^{t}(I,a)u_{\theta}(a)$ 
    \FOR{$a \in A(I)$}
    \STATE $r(I,a)\leftarrow Pr(\theta|h)\cdot u(a)-u_{\sigma^{t}}$
    \ENDFOR
    \STATE Insert the infoset and its action regret values $(I,t,r(I))$ into regret memory $\mathcal{M}_{r,p}$
    \STATE \textbf{return} $u_{\sigma^{t}}$
    \ELSE
    \STATE $I\leftarrow$ Information set containing $h$;
    \STATE $\sigma^{t}(I) \leftarrow$ Strategy of Information set $I$ computed from $R(I,a|\psi_{-p})$ using regret matching $+$;
    \STATE Insert the infoset and its strategy $(I,t,\sigma^{t}(I))$ into strategy memory $\mathcal{M}_{\pi,(-p)}$;
    \STATE Sample an action $a$ from the probability distribution $\sigma^{t}(I)$
    \STATE \textbf{return} DEEPBCFRTRVERSE$(ha,p,\psi_{p},\psi_{-p},\mathcal{M}_{r,p},\mathcal{M}_{\pi,(-p)},\theta)$
    \ENDIF
\end{algorithmic}
\end{algorithm}

\end{proof}

\section{Experment}

\subsection{Experment setting}
Mixed-1=$Normal:conservative:aggressive = 10\%:80\%:10\%$, Mixed-2=$Normal:conservative:aggressive = 20\%:60\%:20\%$,Mixed-3=$Normal:conservative:aggressive = 30\%:40\%:30\%$

Mixed-4=$Normal:conservative:aggressive = 80\%:10\%:10\%$, Mixed-5=$Normal:conservative:aggressive = 60\%:20\%:20\%$,Mixed-6=$Normal:conservative:aggressive = 40\%:30\%:30\%$

Mixed-7=$Normal:conservative:aggressive = 10\%:10\%:80\%$, Mixed-8=$Normal:conservative:aggressive = 20\%:20\%:60\%$,Mixed-9=$Normal:conservative:aggressive = 30\%:30\%:40\%$

\subsection{Ablation studies}

\begin{figure*}[h]
\centering
  \subfigure[Mixed-4]{\includegraphics[width=0.33\textwidth]{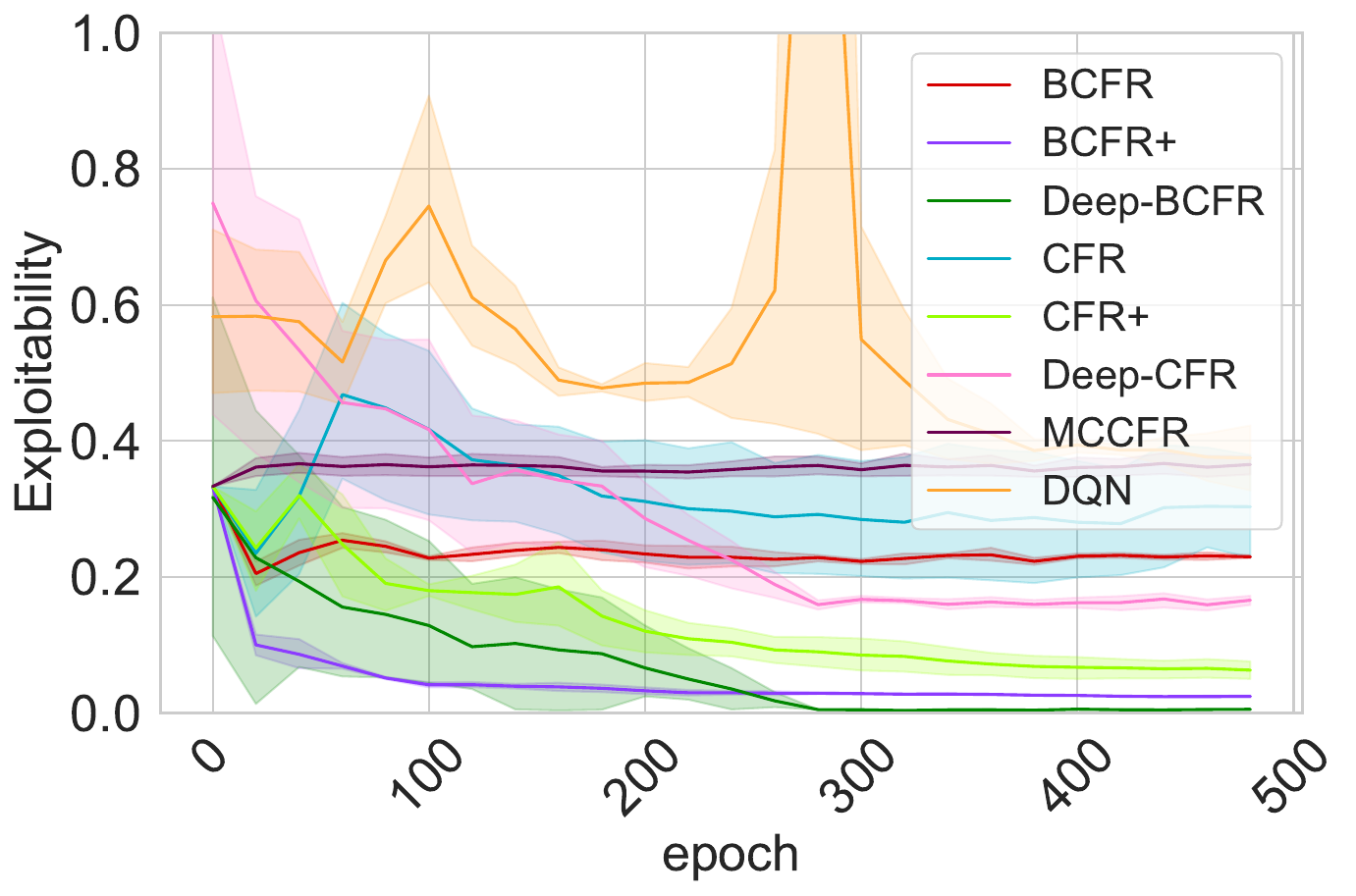}}\hfill
  \subfigure[Mixed-5]{\includegraphics[width=0.33\textwidth]{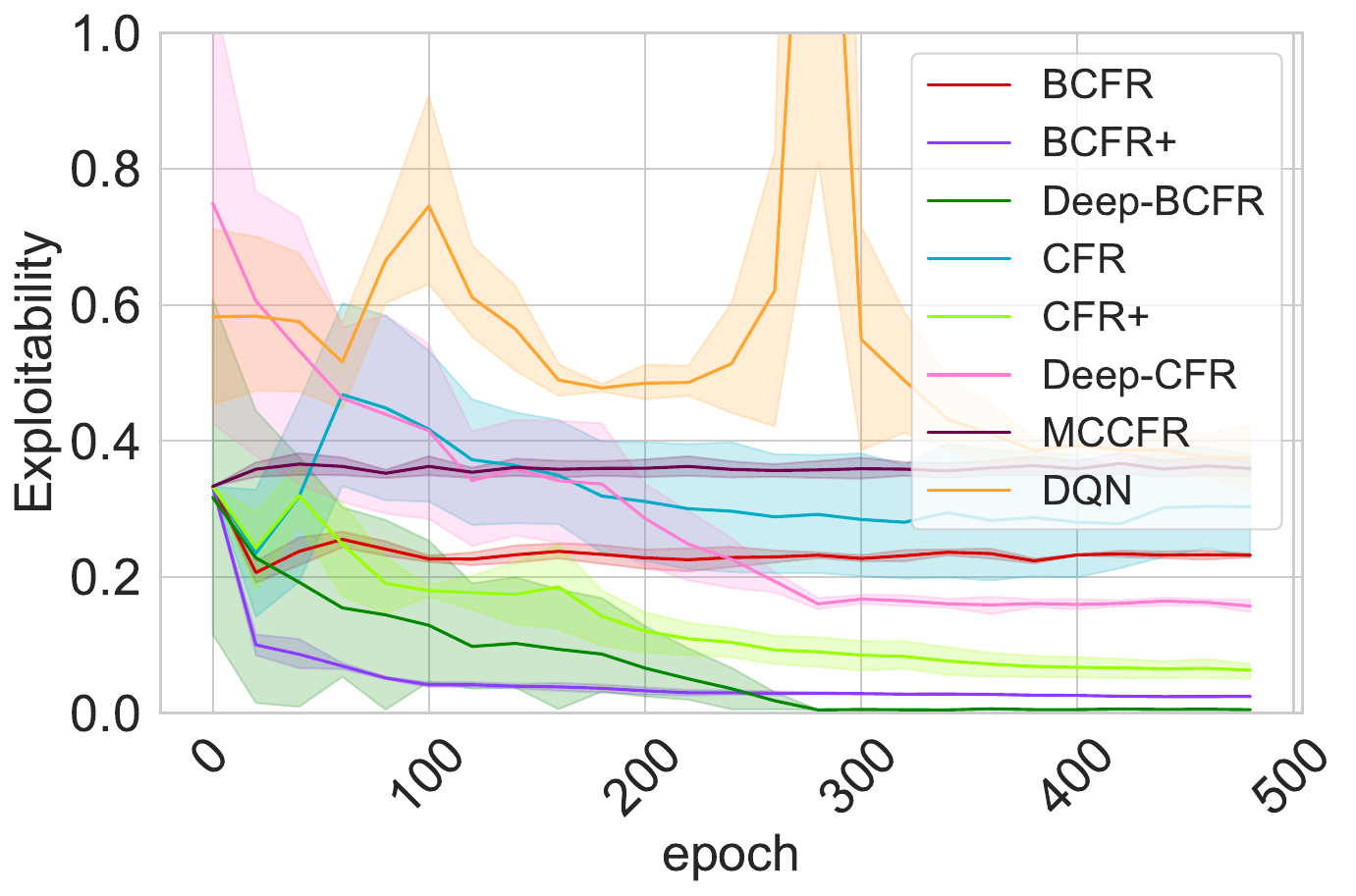}}\hfill
  \subfigure[Mixed-6]{\includegraphics[width=0.33\textwidth]{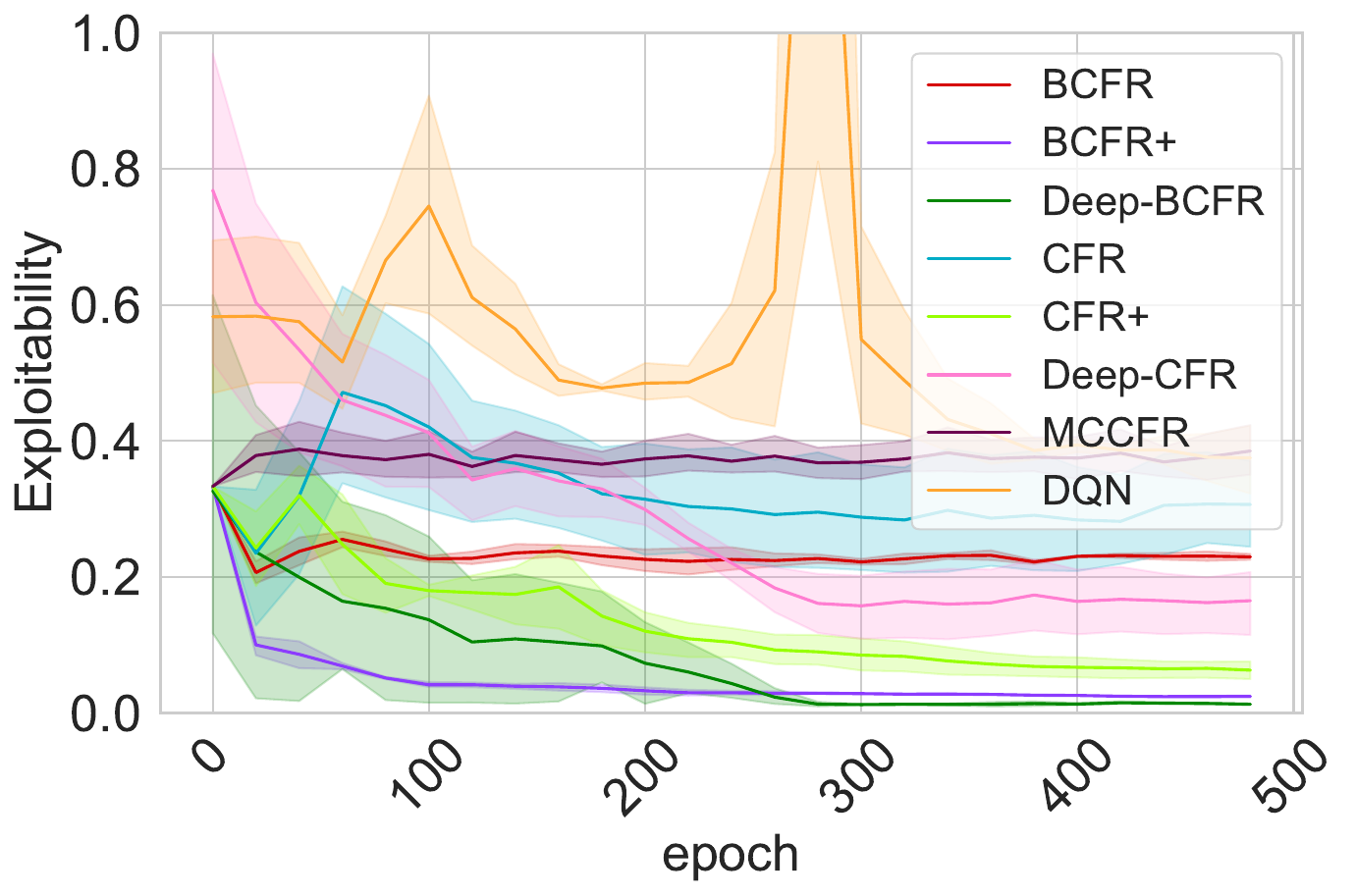}}\hfill
  \\
  \subfigure[Mixed-7]{\includegraphics[width=0.33\textwidth]{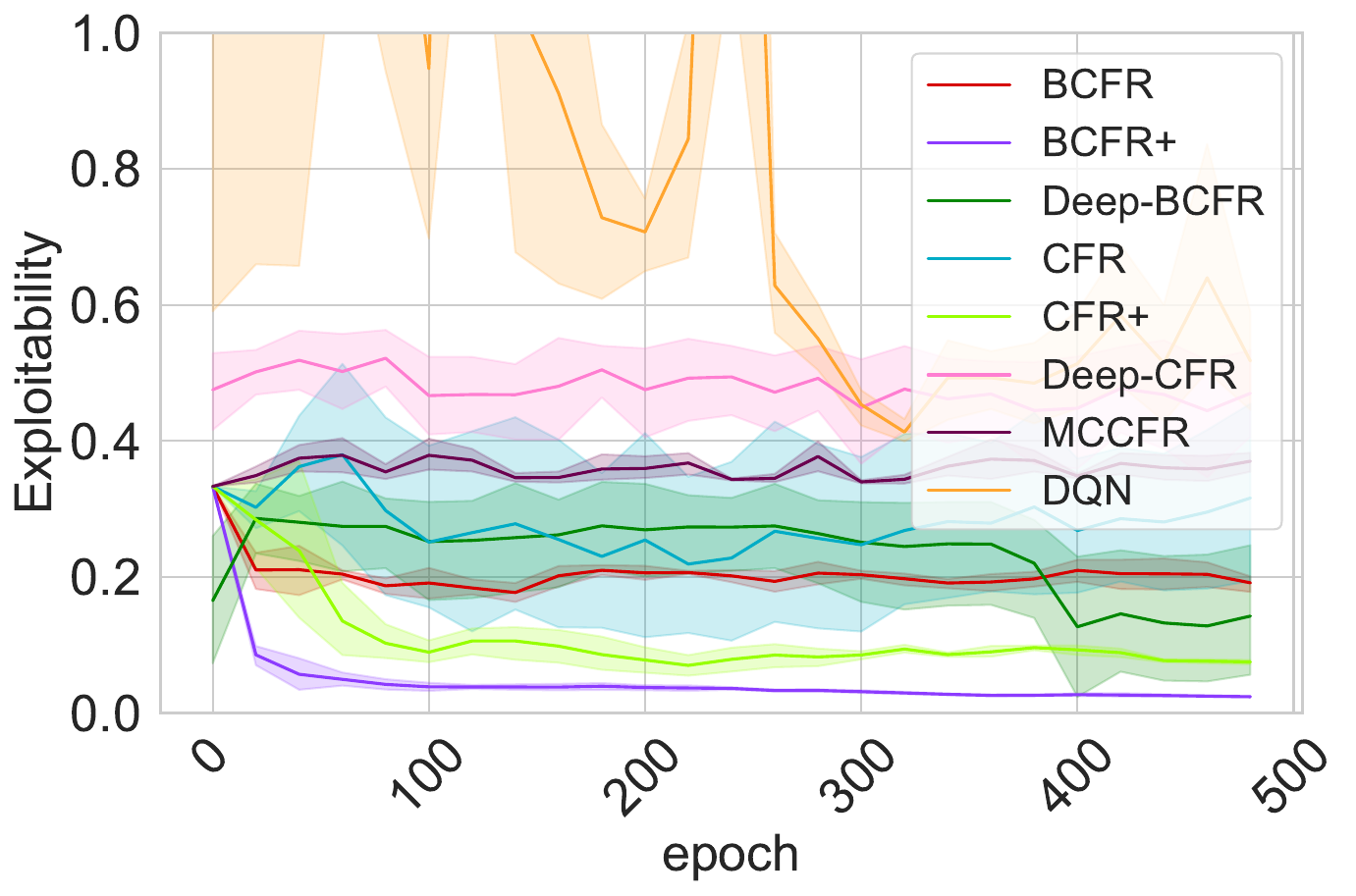}}\hfill
  \subfigure[Mixed-8]{\includegraphics[width=0.33\textwidth]{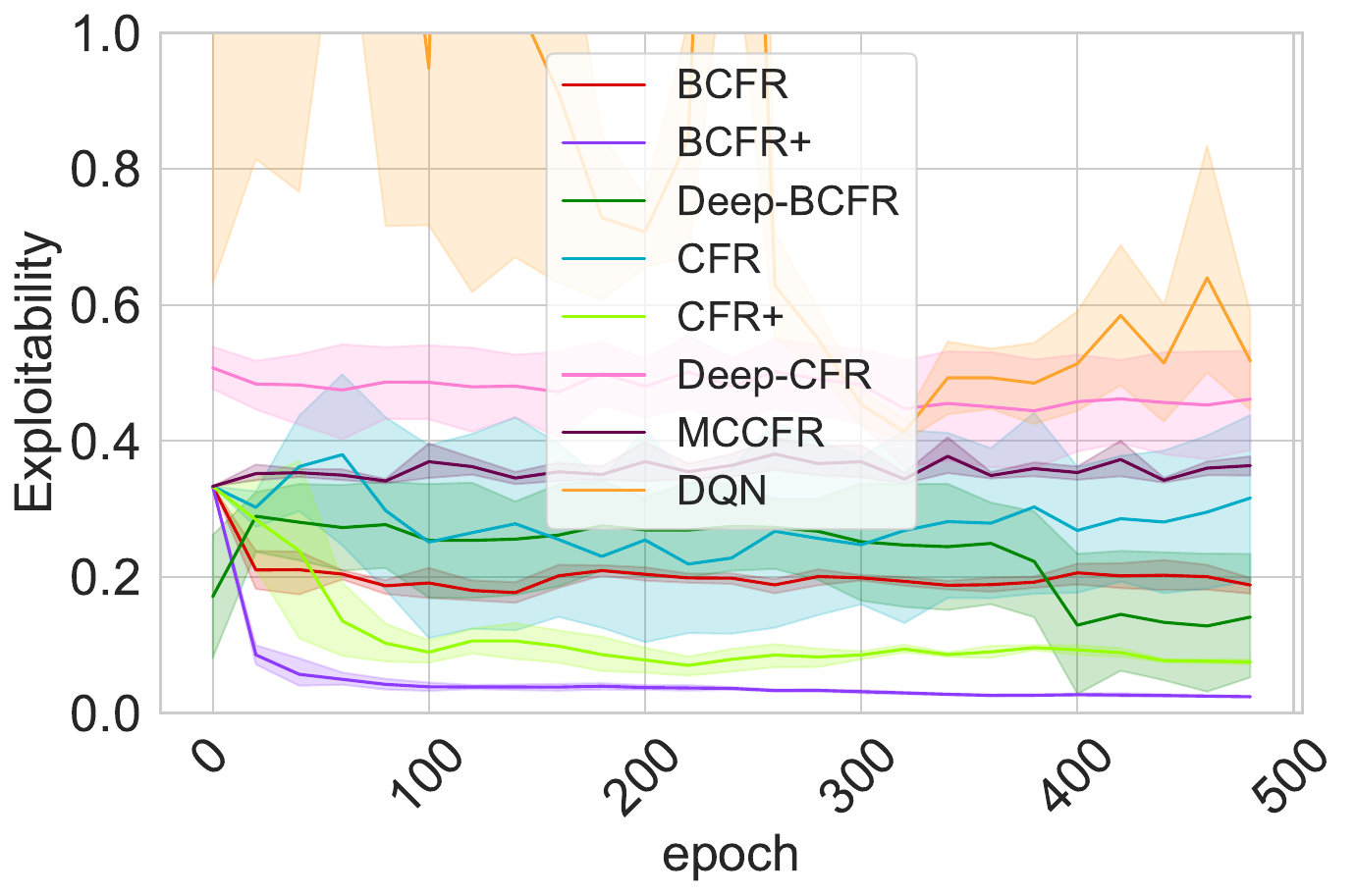}}\hfill
  \subfigure[Mixed-9]{\includegraphics[width=0.33\textwidth]{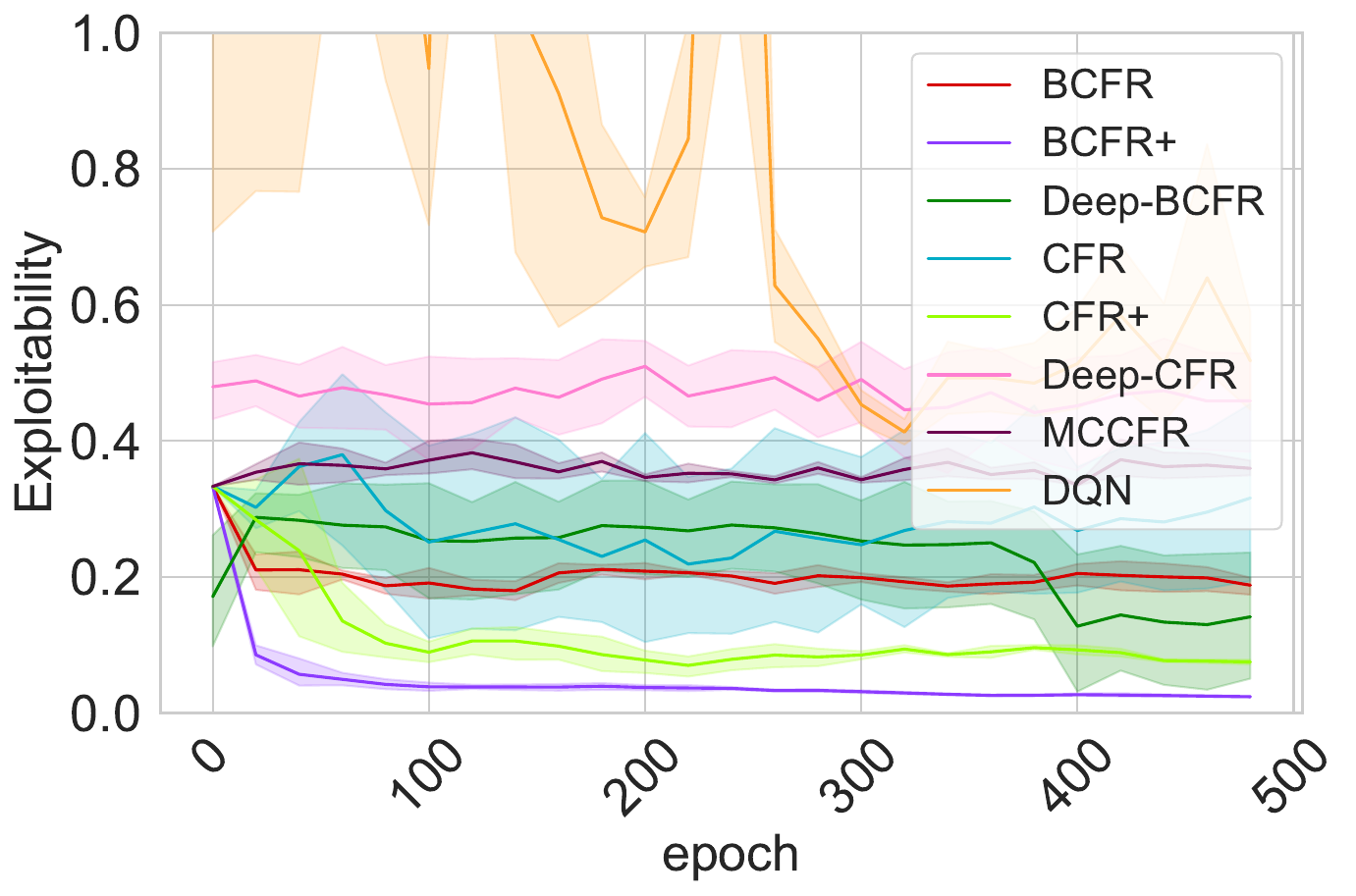}}\hfill
  \caption{Exploitability under three player's types: normal (left), conservative (middle) and aggressive (right).} 
\end{figure*}

To validate the effectiveness of our approach, we designed ablation studies using various methods: CFR, CFR+, and Deep CFR. We conducted experiments using Bayesian-CFR with full knowledge of the player's type. We employed Bayesian-CFR with the player using a uniform distribution and CFR with a known player type. Under uniform distribution of player types, Bayesian-CFR effectively degenerates to using only CFR. Additionally, employing Bayesian-CFR with knowledge of the player's type produced results similar to  CFR with known player types.
In addition, we refer to the following articles~\cite{zhang2024collaborative,zou2024distributed,wang2023osteoporotic,fang2022coordinate,fang2023implementing,zhou2022pac,zhou2023value,mei2022bayesian,mei2024projection,chen2023real,chen2024deep,zhou2023value,chen2021bringing}
\begin{figure*}[h]
\label{fig:Exp2}
\centering
  \subfigure[]{\includegraphics[width=0.33\textwidth]{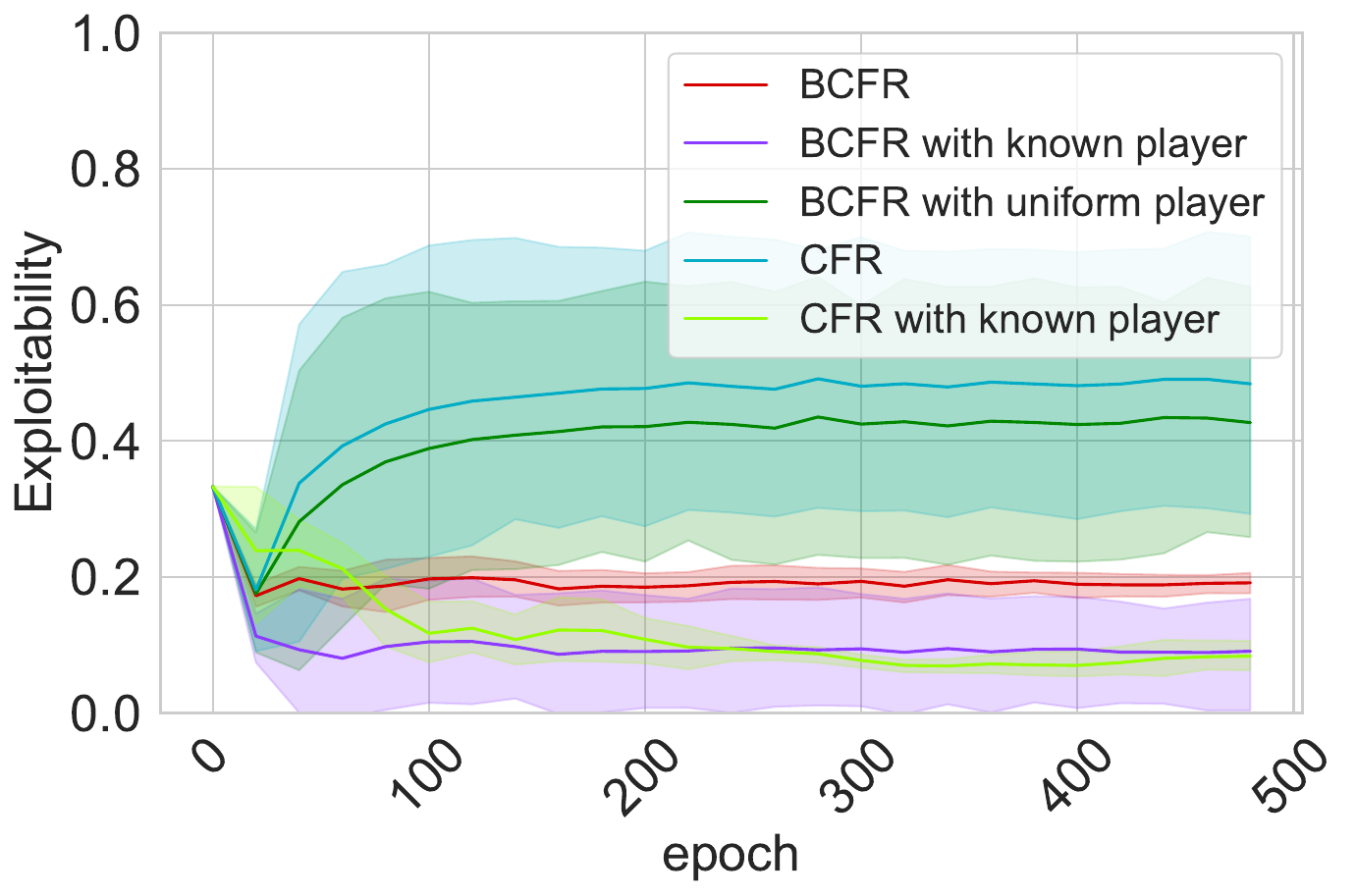}}\hfill
  \subfigure[]{\includegraphics[width=0.33\textwidth]{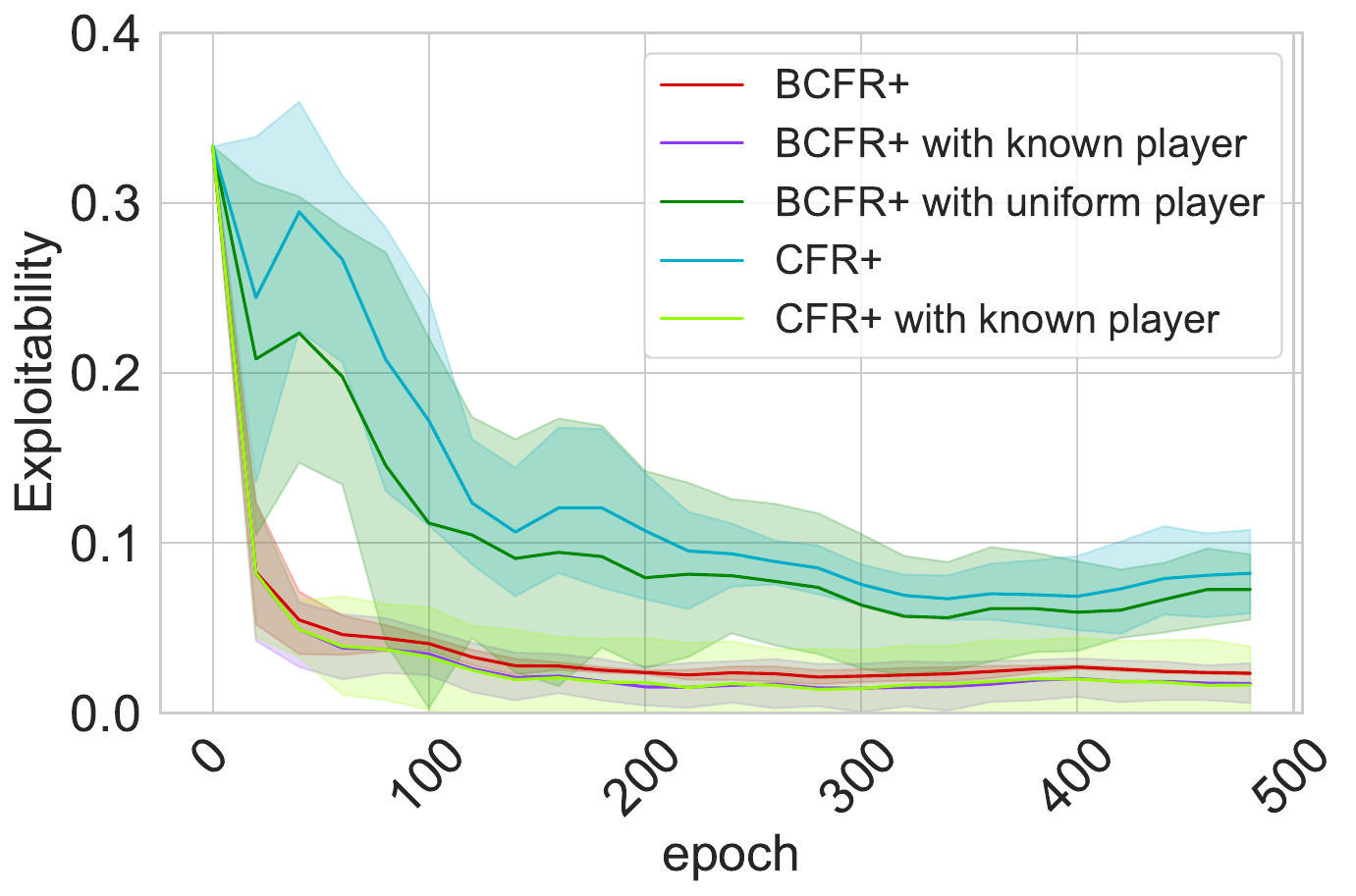}}\hfill
  \subfigure[]{\includegraphics[width=0.33\textwidth]{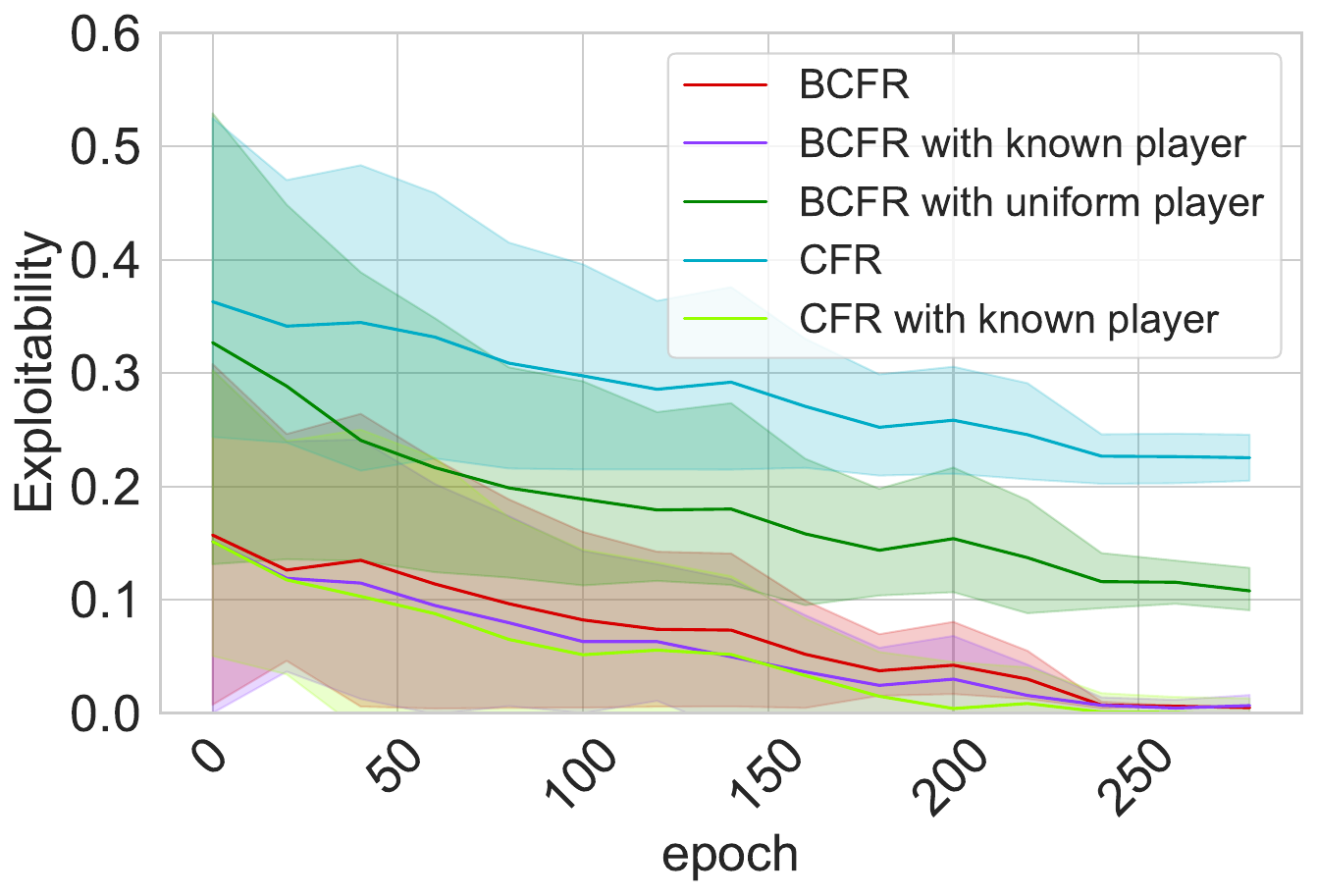}}\hfill
  \caption{This experiment focuses on players of the conservative type. From left to right, the baseline methods used are CFR, CFR+, and Deep CFR. Within each experiment, there are several variations: BCFR, BCFR with knowledge of the player's type, BCFR using uniformly distributed player's type, CFR, and CFR with knowledge of the player's type.}
\end{figure*}